\documentclass[aps,prx,superscriptaddress,twocolumn,twoside,nofootinbib,floatfix,a4paper]{revtex4-2} %
\usepackage[T1]{fontenc}
\usepackage[utf8]{inputenc}
\usepackage{amsmath,amssymb,amsthm,dsfont,bm}
\usepackage{graphicx}
\usepackage{hyperref}
\usepackage{comment}
\usepackage{xr}
\usepackage{bbm}
\usepackage{setspace}
\usepackage{enumitem}
\usepackage{xcolor,colortbl}
\usepackage{multirow}
\usepackage{latexsym}
\usepackage{amsfonts}
\usepackage{mathrsfs}
\usepackage{natbib}
\usepackage{verbatim}
\usepackage{gensymb}
\usepackage{caption}
\usepackage{subcaption}
\usepackage{ragged2e}
\usepackage{oplotsymbl}
\usepackage{array}
\usepackage{physics}
\usepackage{tikz}
\usetikzlibrary{arrows.meta}
\usepackage{tikz-3dplot}
\usepackage{blochsphere}

\usepackage{soul}
\DeclareCaptionJustification{justified}{\justifying}
\usepackage{blkarray}
 \usepackage{mathtools}
\makeatletter
\newtheorem*{rep@theorem}{\rep@title}
\newcommand{\newreptheorem}[2]{%
\newenvironment{rep#1}[1]{%
 \def\rep@title{#2 \ref{##1}}%
 \begin{rep@theorem}}%
 {\end{rep@theorem}}}
\makeatother
\theoremstyle{plain}

\newtheorem*{thm*}{Theorem}
\newreptheorem{thm}{Theorem}

\newtheorem{obs}{Observation}

\newtheorem{defn}{Definition}
\newtheorem*{obs*}{Observation}

\theoremstyle{remark}

\newcommand{\ins}{{\mathsf{I}}}
\newcommand{\insi}{{\mathsf{J}}}

\definecolor{myurlcolor}{rgb}{0,0,0.4}
\definecolor{mycitecolor}{rgb}{0,0,0.4}
\definecolor{myrefcolor}{rgb}{0.5,0,0}
\hypersetup{colorlinks, linkcolor=myrefcolor, citecolor=mycitecolor, urlcolor=myurlcolor}

\newcommand{\one}[0]{\mathds{1}}

\begin{document}
\title{Generalized measurement incompatibility}
 \author{Edwin Peter Lobo}
 \email{edwin.lobo@ulb.be}
 \author{Maria Balanz\'o-Juand\'o}
 \author{Stefano Pironio}
 \affiliation{Laboratoire d'Information Quantique, Universit\'e libre de Bruxelles (ULB), Av. F. D. Roosevelt 50, 1050 Bruxelles, Belgium.}
 
\begin{abstract}
Quantum measurements can be incompatible, i.e., they can fail to be jointly measurable. Recently, a weaker notion of joint-measurability, called partial joint-measurability, was proposed by Masini \textit{et al.} in [Quantum {\bf 8}, 1574 (2024)].
In this work, we further generalize this notion to the setting where only a subset of the outcomes of each measurement is required to be jointly determined by classical variables.
We provide two mathematical formulations of partial joint-measurability and show that, like full joint-measurability, it can be decided by solving a single semidefinite program. 
We prove that in the case of an untrusted measurement device, an adversary Eve, limited to classical side information, can perfectly guess the outcomes of the measurement device if and only if the set of measurements is partially jointly measurable. 
We derive analytical thresholds on the detection efficiency below which generic measurements become partially jointly measurable.
Such bounds directly yield limits on the robustness of device-independent and semi-device-independent quantum cryptographic protocols against detection inefficiency. In particular, our results highlight the importance of a careful treatment of postselection in security analyses.
\end{abstract}

\maketitle
\section{Introduction} 
Measurement incompatibility, i.e., the impossibility of realizing several measurements through classical postprocessing of a single parent measurement, is a necessary ingredient in essentially all demonstrations of nonclassicality in quantum theory, such as Bell nonlocality~\cite{Bell1964,Brunner2014}, quantum steering~\cite{Uola2014,Quintino2014}, and Bell-Kochen-Specker contextuality~\cite{Bell1966,KochenSpecker1967,Mermin1993}. It is also tightly connected to quantum communication tasks with untrusted devices, where it limits the information accessible to an eavesdropper Eve. Conversely, if the measurements performed by an untrusted device are compatible, then Eve can replace it by a purely classical device enabling her to guess the outcome of every measurement; consequently, no fresh randomness is produced and such measurements are useless for tasks such as device-independent (DI) or semi-DI quantum key distribution (QKD)~\cite{Acin2007,Primaatmaja2023,Pawlowski2011,Woodhead2015}.

Recently, a refined notion of joint-measurability, called partial joint-measurability, was introduced in Ref.~\cite{Masini2024}. In most implementations of protocols such as DIQKD, the outcomes of only a subset of measurements are eventually used for key generation. If Eve can perfectly guess these outcomes, there is clearly no security. Partial joint-measurability captures this hybrid notion: the key-generating outcomes are classically determined, and hence available to Eve, while the remaining measurement outcomes may still involve genuine quantum operations. Thus, partial joint-measurability provides a simple attack for quantum communication protocols based on one untrusted measurement device.

In this work, we further generalize partial joint-measurability. In particular, our generalization explicitly accounts for postselection in quantum communication protocols, i.e., the fact that only a subset of outcomes for each input may be used for key generation. We show that deciding whether a given set of measurements is partially jointly measurable can be cast as a single semidefinite program (SDP). We also provide an operational interpretation: an adversary Eve restricted to classical side information (i.e., without quantum memory) can perfectly guess the outcomes of an untrusted measurement device for all input states if and only if the measurements are partially jointly measurable. Furthermore, we derive tight analytical bounds on the threshold detection efficiency at which a set of measurements becomes partially jointly measurable. As an application, we provide a simple attack on the steering-based QKD protocol of Ref.~\cite{Branciard2012}, which invalidates their security proof.

\section{Preliminaries}\label{sec:partialJM}

We consider a scenario where Bob has a measurement device capable of performing $n$ different measurements, labeled by $y \in [n] \coloneqq \{1,\dots,n\}$. Each measurement $y$ has $k_y$ possible outcomes, labeled by $b \in [k_y] \coloneqq \{1,\dots,k_y\}$. For each setting $y$, the measurement is described by a POVM $B_y \coloneqq \{ B_{b|y} \}_{b}$ acting on a Hilbert space $\mathcal{H}$.
\subsection{Operational meaning of joint-measurability}
Let us first recall the notion of standard joint-measurability. 
The measurements $B_y$ are said to be jointly measurable, or compatible, if there exists a single 
`parent measurement' $E = \{E_c\}_c$ and a conditional probability distribution 
$p(b|y,c)$ (often called a response function), such that
\begin{equation}\label{defn:joint-measurability1}
    B_{b|y} = \sum_{c} E_{c}\, p(b|y,c) \,,
\end{equation}
for all measurement choices $y$ and outcomes $b$. Operationally, this means that the measurements $B_y$ can be simulated by a device that 
first performs the single parent measurement $E$, producing a classical outcome $c$, and then 
generates the final output $b$ by classical postprocessing of $c$ conditioned on the chosen 
setting $y$.
\subsection{Partial joint-measurability}
Recently, a weaker notion of joint-measurability, termed {partial joint-measurability}, 
was introduced in Ref.~\cite{Masini2024}. 
In this framework, the outcomes of only a subset of the $n$ measurements are required to be 
fully determined by some classical information $c$ that is independent of $y$, 
while the remaining measurements may still involve genuine quantum operations. 

Here, we propose an even more general definition of partial joint-measurability, 
in which only a subset of the outcomes of each measurement are required to be fully 
determined by the classical information $c$. 
To formalize this, we partition the outcome set $[k_y]$ of each measurement $y$ into two disjoint subsets: 
$\mathcal{G}_y \subseteq [k_y]$ and its complement 
$\overline{\mathcal{G}}_y = [k_y] \setminus \mathcal{G}_y$.

Intuitively, our notion of general partial joint-measurability captures the idea that, conditioned on the event that an outcome 
in $\mathcal{G}_y$ occurs, the specific outcome $b \in \mathcal{G}_y$ is entirely determined by 
some classical information $c$ available prior to the choice of $y$. 
This framework reduces to previous notions of joint-measurability: 
if $\mathcal{G}_y = [k_y]$ for all $y$, we recover standard joint-measurability, 
since every outcome is fully determined by $c$; 
if instead $\mathcal{G}_y = [k_y]$ for some $y$ and $\mathcal{G}_{y'} = \emptyset$ for others, 
we obtain the notion of partial-input joint-measurability introduced in Ref.~\cite{Masini2024}. 
In general, the subsets $\mathcal{G}_y$ may be arbitrary, allowing for a flexible 
and fine-grained characterization of measurement compatibility.

We define this notion formally through a two-step simulation procedure. 
First, consider a general quantum operation that is independent of $y$ and produces a classical outcome $c$ together with an output state $\rho_c$. In full generality, this corresponds to a quantum instrument 
$\ins=\{\ins_c\}_c$, i.e., a collection of completely positive, 
trace-nonincreasing maps such that $\sum_c \ins_c$ is trace-preserving. 
Acting on an input state $\rho$, the instrument outputs a classical label $c$ together with the (unnormalized) post-measurement state $\ins_c(\rho)$. 
Next, we allow for a $y$-dependent operation that takes as input both $c$ and the state $\ins_c(\rho)$, and produces a classical outcome $b$. 
This is described by a collection of POVMs 
$M_{y,c}=\{M_{b|y,c}\}_{b}$, see Fig.~\ref{fig:effective_povm}. We can now formally define partial joint-measurability.

\begin{defn}\label{defn:partial-joint-measurability}
Given a collection of subsets $\mathcal{G} = (\mathcal{G}_y)_y$ with 
$\mathcal{G}_y \subseteq [k_y]$, the measurements 
$B_y = \{B_{b|y}\}_{b}$ are said to be \emph{$\mathcal{G}$-jointly measurable} ($\mathcal{G}$-JM)  if there exists 
a quantum instrument $\ins=\{\ins_c\}_c$ and a family of POVMs 
$M_{y,c}=\{M_{b|y,c}\}_{b}$ such that, for all $y \in [n]$ 
\begin{align}
    \Tr[\rho B_{b|y}] &= \sum_c \Tr[\ins_{c}(\rho)\, M_{b|y,c}] \quad \forall  \rho,  b \in [k_y]\,,\label{eq:right_prob}\\
	 M_{b|y,c} &= p(b|y,c)\, M_{\star|y,c} \quad \forall b \in \mathcal{G}_y \,, \label{eq:partial-JM-constraint}
\end{align}
for some positive operator 
$
M_{\star|y,c} := \one - \sum_{b \in \overline{\mathcal{G}}_y} M_{b|y,c}
$
independent of $b$, and where $p(b|y,c)$ is a conditional probability distribution satisfying $p(b|y,c)\ge 0$ and $\sum_{b \in \mathcal{G}_y} p(b|y,c) = 1$.
\end{defn}

In this definition, condition~\eqref{eq:right_prob} guarantees that the overall simulation is 
indistinguishable from the original measurements $B_y$ for all input states $\rho$. 
Condition~\eqref{eq:partial-JM-constraint} enforces that, for outcomes $b \in \mathcal{G}_y$, 
the operators $M_{b|y,c}$ are proportional to a common operator $M_{\star|y,c}$, 
with proportionality given by the classical response function $p(b|y,c)$. 
Thus, conditioned on obtaining an outcome in $\mathcal{G}_y$, the specific $b$ is determined 
entirely by the classical value $c$. 
Operationally, we can interpret this as follows: after the action of the instrument $\ins$, 
a quantum measurement (depending on both $y$ and $c$) is performed on $\ins_c(\rho)$, with 
POVM elements $\{M_{b|y,c}\}_{b \in \overline{\mathcal{G}}_y} \cup \{M_{\star|y,c}\}$. Here $M_{\star|y,c}$ corresponds to the event “an outcome in $\mathcal{G}_y$ occurs”, 
in which case the precise outcome $b \in \mathcal{G}_y$ is sampled classically 
according to $p(b|y,c)$.

If $\mathcal{G}_y = [k_y]$ for all $y$, then it is easy to see that Def.~\ref{defn:partial-joint-measurability} recovers the definition of standard joint-measurability. Indeed, if $\mathcal{G}_y = [k_y]$ for all $y$, then $M_{\star|y,c} = \one$ according to Def.~\ref{defn:partial-joint-measurability}, and the POVM elements of every measurement performed by Bob's device are of the form $M_{b|y,c} = p(b|y,c) \one$. Thus, Bob's device is a purely classical device that performs classical postprocessing conditioned on $y,c$, and the role of the quantum instrument $\{\ins_c\}_c$ can be restricted to supplying the classical outcome $c$ to Bob's device. This can be modelled as a measurement $\{E_c\}_c$ and we can rewrite Eq.~\eqref{eq:right_prob} as $\Tr[\rho B_{b|y}] = \sum_c p(b|y,c) \Tr[\rho\,E_{c}]\, \forall  \rho,  b \in [k_y]$, which is equivalent to Eq.~\eqref{defn:joint-measurability1}. If instead $\mathcal{G}_y = [k_y]$ for some $y$ and $\mathcal{G}_{y'} = \emptyset$ for others, then the Def.~\ref{defn:partial-joint-measurability} recovers the notion of partial-input joint-measurability introduced in Ref.~\cite{Masini2024}.

\begin{figure}
    \includegraphics[width=0.48\textwidth]{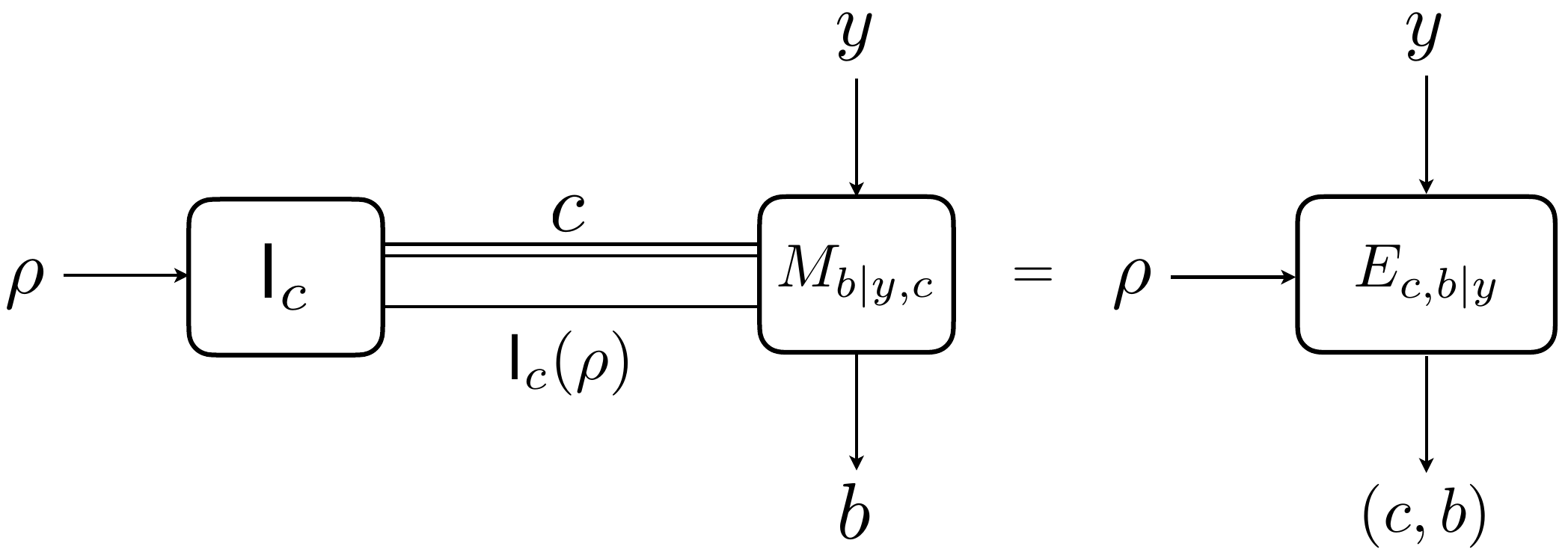}
    \caption{Left:
	The overall simulation of the measurements $B_y$ consists of first applying the instrument $\{\ins_c\}_c$ to $\rho$, obtaining $c$ and $\ins_c(\rho)$, and then performing the measurement $M_{y,c}$ on $\ins_c(\rho)$ to obtain $b$. Right: The combination of the quantum instrument $\ins=\{\ins_{c}\}_c$ and the measurement $M_{y,c}=\{{M}_{b|y,c}\}_{b}$ can be seen as a single effective POVM $E_y=\{E_{c,b|y}\}_{c,b}$ with input $y$ and outcomes $(c,b)$.}
    \label{fig:effective_povm}
\end{figure}
We can view the combination of the instrument $\ins$ and the measurement $M_{y,c}$ in Def.~\ref{defn:partial-joint-measurability} as a single effective measurement $E_y=\{E_{c,b|y}\}_{c,b}$ with input $y$ and outcomes $(c,b)$, see Fig.~\ref{fig:effective_povm}. Indeed, using the adjoint maps $\ins_c^\dagger$, we can write 
\begin{equation}
	\Tr[\ins_c(\rho) M_{b|y,c}] = \Tr[\rho \,\ins_{c}^\dagger({M}_{b|y,c})] = \Tr[\rho E_{c,b|y}] \,,
\end{equation}
where we have defined the effective measurement operators,
\begin{equation}
	E_{c,b|y} \coloneqq \ins_{c}^\dagger ({M}_{b|y,c})\,.
\end{equation}
 Since the adjoint maps $\ins_{c}^\dagger$ are completely positive and $\sum_c \ins_c^\dagger$ is unital, i.e., $\sum_c \ins_c^\dagger(\one) = \one$, the operators $E_{c,b|y}$ form valid POVMs: they are positive and satisfy $\sum_{c,b} E_{c,b|y}=\one$ for all $y$. Further, they also obey the no-signaling conditions
\begin{equation}\label{eq:no-signaling}
	\sum_{b} E_{c,b|y}=E_c \quad \forall y \in [n]\quad\text{(no-signaling)}\,,
\end{equation}
i.e., marginalizing over the outcomes $b$ yields an operator $E_c$ that is independent of $y$. This follows immediately from $\sum_b E_{c,b|y} = \sum_b \ins_{c}^\dagger(M_{b|y,c}) = \ins_{c}^\dagger(\one)\eqqcolon E_c$.
In terms of these effective measurements, the conditions~\eqref{eq:right_prob} and~\eqref{eq:partial-JM-constraint} become
\begin{align}
	&B_{b|y} = \sum_c E_{c,b|y}&\text{(consistency)}\,, \label{eq:consistency} \\
	&E_{c,b|y} = p(b|y,c)\, E_{c,\star|y} \quad \forall b \in \mathcal{G}_y &\text{(partial JM)}\,,\label{eq:partial-JM}
\end{align}
for some positive operator $E_{c,\star|y} = E_c - \sum_{b \in \overline{\mathcal{G}}_y} E_{c,b|y}$.

Conversely, any set of POVMs $E_{y}=\{{E}_{c,b|y}\}_{b,c}$ that satisfies the conditions~\eqref{eq:no-signaling}--\eqref{eq:partial-JM} can be interpreted in the sense of Def.~\ref{defn:partial-joint-measurability}.  That is, given such operators $E_{c,b|y}$, one can always construct a quantum instrument $\ins=\{\ins_c\}_c$ and measurements $M_{y,c}=\{M_{b|y,c}\}_{b}$ satisfying the conditions of Def.~\ref{defn:partial-joint-measurability}, as shown in App.~\ref{app:equivalence_PJM}. Hence, partial joint-measurability can equivalently be defined as the existence of POVMs $E_y=\{E_{c,b|y}\}_{c,b}$ that satisfy conditions~\eqref{eq:no-signaling}--\eqref{eq:partial-JM}. We call these operators partial parent (PP) POVMs, since they play a similar role to the parent POVM in full joint-measurability. Indeed, if $\mathcal{G}_y = [k_y]$ for all $y$, then the partial JM condition~\eqref{eq:partial-JM} implies that $E_{c,b|y} = p(b|y,c)\, E_c$ for all $b,y$ and all the effective measurements arise from a single parent POVM $E=\{E_c\}_c$ followed by classical postprocessing. 

\subsection{Reformulation of partial joint-measurability as an SDP}\label{sec:sdp_reformulation}
In the case of full joint-measurability, one can always assume, without loss of generality, that the response function $p(b|y,c)$ is deterministic, i.e., $p(b|y,c)\in\{0,1\}$ for all $b,y,c$. The same holds for the definition of $\mathcal{G}$-joint-measurability. More specifically, we can assume that the outcome $c$ is a tuple $\bar{\beta} = (\beta_1,\ldots,\beta_n)$ with $\beta_y \in \mathcal{G}_y$ (and $\beta_y=\emptyset$ if $\mathcal{G}_y = \emptyset$), which deterministically specifies that the outcome $b=\beta_y$ should be output whenever the input is $y$ and the outcome lies in $\mathcal{G}_y$. 
This leads to the following equivalent formulation of $\mathcal{G}$-joint-measurability.
\begin{defn}\label{defn:partial-joint-measurability2}
Given a collection of subsets $\mathcal{G} = (\mathcal{G}_y)_y$ with 
$\mathcal{G}_y \subseteq [k_y]$, the measurements 
$B_y = \{B_{b|y}\}_{b}$ are said to be \emph{$\mathcal{G}$-jointly measurable} if there exist 
POVM elements $E_{\bar{\beta},b|y}$ such that
\begin{align}
	&\sum_{b} E_{\bar{\beta},b|y} = E_{\bar{\beta}} \quad \forall y \in [n]\,, &\text{(no-signaling)}\,,\label{eq:no-signalingSDP}\\
	&B_{b|y} = \sum_{\bar{\beta}} E_{\bar{\beta},b|y}\,,&\text{(consistency)}\,,  \label{eq:consistencySDP}\\
	&E_{\bar{\beta},b|y} = \delta_{b,\beta_y}\,E_{\bar{\beta},\star|y} \quad \forall b \in \mathcal{G}_y\,&\text{(partial JM)}\,, \label{eq:partial-JM_SDP}
\end{align}
for some positive operator $E_{\bar{\beta},\star|y} := E_{\bar{\beta}} - \sum_{b \in \overline{\mathcal{G}}_y} E_{\bar{\beta},b|y}$ independent of $b$. 
\end{defn}
It is clear that the above formulation is a special case of the general one in terms of arbitrary PP POVMs $E_{c,b|y}$. Conversely, as shown in App.~\ref{app:proof:deterministic_response}, starting from any PP POVMs $E_{c,b|y}$ one can always find  POVMs $E_{\bar{\beta},b|y}$ satisfying Eqs.~\eqref{eq:no-signalingSDP}--\eqref{eq:partial-JM_SDP}. Hence, the two definitions are equivalent.
The second one is particularly useful: since the number of operators $E_{\bar{\beta},b|y}$ is finite, deciding whether they can be compatible with Eqs.~\eqref{eq:no-signalingSDP}--\eqref{eq:partial-JM_SDP}, i.e., whether the measurements $B_y$ are $\mathcal{G}$-JM, can be formulated as an SDP. It also leads to the following Observations.

\begin{obs}\label{obs:1}
	If the measurements $\{B_y\}_{y\in[n]}$ are  $\mathcal{G}$-JM, then they are also $\mathcal{G}'$-JM whenever $\mathcal{G}'$ satisfies $\mathcal{G}'_y\subseteq \mathcal{G}_y$ for all $y\in[n]$.
\end{obs}
\begin{proof}
This is immediate since shrinking each set $\mathcal{G}_y$ to a subset only relaxes the constraints~\eqref{eq:partial-JM_SDP}, while leaving the consistency and no-signaling conditions unchanged.
\end{proof}

\begin{obs}\label{obs:2}
	If the measurements $\{B_y\}_{y\in [n]}$ are  $\mathcal{G}$-JM, then for any subset $\mathcal{U} \subseteq [n]$, the measurements $\{B_y\}_{y\in \mathcal{U}}$ are $\mathcal{G}'$-JM, where $\mathcal{G}' = (\mathcal{G}_y)_{y\in\mathcal{U}}$.
\end{obs}
\begin{proof}
This is also immediate since we only relax the constraints in Def.~\ref{defn:partial-joint-measurability2}.
\end{proof}

\begin{obs}[Masini \emph{et al.}~\cite{Masini2024}]\label{obs:3}
	If the measurements $\{B_y\}_{y\in [n]}$ are  $\mathcal{G}$-JM, where $\mathcal{G}_y = [k_y]$ for all $y \neq y'$ and $\mathcal{G}_{y'} = \emptyset$ for a single input $y' \in [n]$, then the measurements $\{B_y\}_{y \in [n]}$ are fully jointly measurable.
\end{obs}
\begin{proof}
This observation was made in Ref.~\cite{Masini2024}, following a similar argument first presented in Ref.~\cite{Acin2016}. We reproduce the proof in App.~\ref{app:proof:obs3} for completeness.
\end{proof}

\section{Operational interpretation of partial joint-measurability} \label{sec:Eve_attack}

The above characterization also admits a clear operational interpretation in the context 
of quantum cryptography. Consider a scenario in which Bob's measurement device is untrusted 
and an adversary, Eve, has access to the input quantum channel 
of Bob's device. This setting is standard in DI or semi-DI QKD. 

If Bob's honest measurements $B_y$ are $\mathcal{G}$-JM, then
there exists a quantum instrument $\ins=\{\ins_{\bar{\beta}}\}_{\bar{\beta}}$ and POVMs $M_{y,\bar{\beta}}=\{M_{b|y,\bar{\beta}}\}_b$ that reproduce for any input state $\rho$ the same statistics as the original measurements $B_y$, but which returns the outcome $b=\beta_y$ with certainty whenever $b \in \mathcal{G}_y$.
Since Bob's device is untrusted, we cannot exclude the possibility that Bob's device is implementing the POVMs $M_{y,\bar{\beta}}$ instead of $B_y$ (Eve could have replaced the original device with a malicious one). But then Eve can mount the following attack: she applies the instrument $\ins$ to the input state $\rho$ and records a copy of the classical outcome $\bar{\beta}$. She then forwards the post-measurement state $\ins_{\bar{\beta}}(\rho)$ to Bob's device, together with the classical information $\bar{\beta}$. Upon receiving Bob's input $y$, the device performs the measurement $M_{y,\bar{\beta}}$. If $b\in\mathcal{G}_y$, Eve then knows with certainty that the outcome is $b=\beta_y$. In other words, $\mathcal{G}$-JM implies that there is no intrinsic randomness in the subset of outcomes specified by $\mathcal{G}$, and hence Eve can perfectly guess these outcomes.

A converse statement also holds. In the setting of an untrusted measurement device, if an adversary Eve—limited to classical side information (i.e., without access to quantum memory)—can perfectly guess the outcomes $b$ of the measurements $B_y$ whenever $b \in \mathcal{G}_y$ for any input state $\rho$,
then the measurements $B_y$ must necessarily be $\mathcal{G}$-JM.
 
To see this, note that the most general strategy for an adversary Eve without quantum 
memory corresponds to the procedure depicted in Fig.~\ref{fig:effective_povm}. 
Eve applies a quantum instrument $\ins=\{\ins_c\}_c$ to the input state $\rho$, 
producing a classical outcome $c$ and a post-measurement state $\ins_c(\rho)$. 
She forwards $\ins_c(\rho)$ and $c$ to Bob's device, while keeping a copy of $c$. 
Upon receiving Bob's input $y$, the device performs a measurement 
$M_{y,c}=\{M_{b|y,c}\}_b$ that was preprogrammed by Eve, as the device is untrusted. Eve produces a guess $g\in \mathcal{G}_y$ for Bob's outcome $b$ according to some distribution $q(g|y,c)$, which depends only on $y$ and $c$ since this is her only side information.

Eve's attack defines effective POVMs $E_y=\{E_{c,b|y}\}_{c,b}$ that satisfy by construction the no-signaling conditions~\eqref{eq:no-signaling} and the consistency conditions~\eqref{eq:consistency} (otherwise, Eve's attack would not reproduce the correct statistics for all input states $\rho$). It remains to show that they also satisfy the partial JM condition~\eqref{eq:partial-JM}.

Fix $y$ and $c$. Note that the joint probability that the side information $c$ is produced, Bob's device returns 
outcome $b$, and Eve outputs guess $g$, given input $y$, factors as
\begin{align}
	p(c,b,g|y) = \Tr[\rho\, E_{c,b|y}]\, q(g|y,c)\,.
\end{align}
Pick some $g'$ such that $q(g'|y,c)$ is non-zero. Then we must have $p(c,b,g'|y) = 0$ for all $b \in \mathcal{G}_y$ such that $b\neq g'$, otherwise Eve would not be able to guess Bob's outcome perfectly whenever $b \in \mathcal{G}_y$. But because of the above factorization, this implies that either (i) $E_{c,b|y}$ is zero for all $b \in \mathcal{G}_y$, or (ii) $E_{c,b|y}$ is non-zero for the single value $b=g'$. In the first case, the partial JM condition~\eqref{eq:partial-JM} is satisfied with $E_{c,\star|y}=0$. In the second case, it is satisfied with $E_{c,b|y} = \delta_{b,g'} E_{c,\star|y}$ where $E_{c,\star|y} = E_{c,g'|y}$.

\section{Illustrations and applications}\label{sec:applications}

We now illustrate our generalized notion of $\mathcal{G}$-joint-measurability and its operational interpretation by applying it to several examples.
We focus on the setting of photonic implementations of quantum communication protocols, where losses are unavoidable and a measurement device may return a `no-click' outcome $\varnothing$ corresponding to the non-detection of the photon. In this case, the POVM elements of an ideal measurement $B_y=\{B_{b|y}\}_b$ are modified as
\begin{equation}\label{eq:lossy_measurements}
	B^{\eta}_{b|y} = 
	\begin{cases}
		\eta \, B_{b|y}  &\text{if} \quad b \neq \varnothing\,,\\
		(1-\eta) \one  &\text{if} \quad b = \varnothing\,,
	\end{cases}
\end{equation}
where $\eta$ is the probability of Bob's device detecting the quantum particle, commonly known as the detection efficiency. 
For values of $\eta$ sufficiently low, the effective measurements $B^{\eta}_y$ become partially jointly measurable, and hence useless for DI and semi-DI applications, as explained below.

\subsection{Generic strategies}\label{sec:generic_strategies}
We first consider a scenario where Bob's device performs $n$ arbitrary measurements $y \in \{1,\dots,n\}$ of the form \eqref{eq:lossy_measurements}, each with $k+1$ outcomes $b \in \{1,\dots,k,\varnothing\}$. We identify several generic strategies---independent of the specific form of the POVMs $B_y$---that render the lossy measurements $B^\eta_y$ $\mathcal{G}$-jointly measurable whenever $\eta$ is below a certain threshold. We analyse four cases, defined by the sets $\mathcal{G}_y$, and summarized in the following table.
\begin{table}[h]
    \centering
    \small
    \renewcommand{\arraystretch}{1.8}
    \begin{tabular}{l|c|c}
       \multicolumn{1}{c|}{JM}  & Subsets $\mathcal{G}$ & Bound on $\eta$ for $\mathcal{G}$-JM \\ \hline
       	\text{(\emph{a})} Full 
            & {$\mathcal{G}_y = \{1,\ldots,k,\varnothing\}\;\forall\, y$}
            & $\eta\leq \frac{1}{n}$ \\ \hline
        \text{(\emph{b})} Partial-input
            & $\begin{array}{@{}c@{}}\mathcal{G}_1 = \{1,\ldots,k,\varnothing\}\\[-4pt] \mathcal{G}_y = \emptyset \;(y\neq 1)\end{array}$
            & $\eta\leq \frac{1}{2}$ \\ \hline
        \text{(\emph{c})} Partial-outcome
            & {$\mathcal{G}_y = \{1,\ldots,k\}\;\forall\, y$}
            & $\eta\leq \max\{\frac{1}{n}, \frac{1}{k}\}$ \\ \hline
		$\begin{array}{@{}l@{}}\text{(\emph{d})} \text{ Partial input}\\[-4pt] \quad\text{\& outcome}\end{array}$
			& $\begin{array}{@{}c@{}}\mathcal{G}_1 = \{1,\ldots,k\}\\[-4pt] \mathcal{G}_y = \emptyset \;(y\neq 1)\end{array}$
			& $\eta\leq \frac{k}{2k-1}$	
    \end{tabular}
	\caption{Detection efficiencies $\eta$ for which the measurements $B^\eta_y$ are $\mathcal{G}$-JM using generic strategies.}\label{tab:cases}
\end{table}

\vspace{-6pt}
\subsubsection*{(a) Full JM} In this regime, Eve aims to guess the outcomes of all measurements, including no-click events. As follows from Ref.~\cite{Acin2016}, the measurements $B^\eta_y$ are fully jointly measurable whenever $\eta \le 1/n$. A strategy achieving this bound is as follows: Eve randomly guesses which measurement $y'$ will be chosen by Bob and performs $B_{y'}$ on the incoming state. She then forwards her guess $y'$ and the resulting outcome $b$ to Bob's device. If Bob's actual input $y$ matches $y'$ (with probability $1/n$), the device outputs $b$; otherwise, it outputs $\varnothing$. In the framework of Def.~\ref{defn:partial-joint-measurability} (or Eq.~\eqref{defn:joint-measurability1}), this is described by a parent POVM with elements $E_c=E_{(y',b')} = B_{b'|y'}/n$ and a response function $p(b|y,c) = p(b|y,(y',b'))= \delta_{b,b'} \delta_{y,y'} + \delta_{b,\varnothing} (1-\delta_{y,y'})$. This bound is known to be tight, as there exist sets of $n$ measurements in every dimension $d\ge2$ that are not compatible for $\eta > 1/n$~\cite{Skrzypczyk2015}.

\vspace{-6pt}
\subsubsection*{(b) Partial-input JM} In this case, Eve is interested in the outcomes of only a single measurement, here $y=1$, a situation common in QKD protocols where one setting is used for key generation and others for device testing. While the bound $\eta \le 1/n$ still applies, a more refined strategy from Ref.~\cite{Acin2016} improves the threshold to $\eta \le 1/2$. In this strategy, Eve performs $B_1$ with probability $1/2$, and forwards the outcome to Bob. With the remaining probability $1/2$, she simply forwards the state unaltered. If Bob's input is $y=1$ and Eve also measured $B_1$, Bob's device outputs Eve's outcome; if $y \neq 1$ and Eve did not measure $y=1$ and instead forwarded the state unaltered, Bob's device performs the honest measurement $B_y$ and outputs the outcome. In all other cases the device outputs $\varnothing$. This strategy can again be described in the framework of Def.~\ref{defn:partial-joint-measurability}. The instrument $\ins$ consists of $k+1$ elements $\ins_c$ where $c\in\{1,\ldots,k,\varnothing\}$. The corresponding Kraus operators are $K_c=\sqrt{1/2} \, \sqrt{B_{c|1}}$ for $c \in \{1,\ldots,k\}$ and $K_{\varnothing} = \sqrt{1/2} \, \one$. The measurement operators $M_{b|y,c}$ are given by $M_{b|y=1,c} = \delta_{c,b}  \, \one$ for $y = 1$, and
\begin{equation} 
	\begin{aligned}
		\forall y \neq 1\,, \quad &M_{b|y,c} = \begin{cases}
		\delta_{c,\varnothing} \, B_{b|y} &\text{if} \quad b \neq \varnothing\,,\\
		(1- \delta_{c,\varnothing}) \, \one &\text{if} \quad b = \varnothing\,.
		\end{cases}
	\end{aligned}
\end{equation}
This bound is also tight: it suffices to consider $n=2$. In this case, partial-input and full joint-measurability coincide (see Observations~\ref{obs:2} and~\ref{obs:3}), and the limit $\eta = 1/2$ is saturated by known incompatible measurements~\cite{Skrzypczyk2015}.

\vspace{-6pt}
\subsubsection*{(c) Partial-outcome JM}
While the notions of joint-measurability for cases \emph{(a)} and \emph{(b)} have been considered in previous works, and particularly in Ref.~\cite{Masini2024}, case \emph{(c)} requires the generalization of partial joint-measurability introduced in the present work. This case is relevant in situations with postselection~\cite{Branciard2012,Thinh2016,Xu2022} where only the conclusive outcomes are kept. For instance, in certain DI or semi-DI QKD protocols, the raw key is generated only from the conclusive outcomes while the no-click outcomes are discarded in key generation rounds. No-click events are still retained and taken into account in testing rounds. In such a scenario, Eve is interested in guessing only the conclusive outcomes of the measurements, which she can do perfectly whenever the measurements $B^\eta_y$ are $\mathcal{G}$-JM for the sets $\mathcal{G}_y = \{1,\ldots,k\}$.

From Observation~\ref{obs:1}, the bound of $\eta \le 1/n$ for full joint-measurability applies to this case as well (though it is not necessarily tight). However, a different strategy proposed in Ref.~\cite{Chaturvedi2024} can be implemented, in which Eve, instead of guessing which input is going to be used by Bob, guesses which output is going to be obtained by Bob, resulting in the alternative bound $\eta \le 1/k$. In this strategy, Eve guesses the \emph{outcome} rather than the \emph{input}: she guesses an outcome $c \in \{1,\dots,k\}$ with probability $1/k$. Bob's device then performs $B_y$; if the outcome is $b=c$, it is output; otherwise, the device returns $\varnothing$. Formally, this corresponds to an instrument $\ins$ with Kraus operators $K_c = \sqrt{1/k} \one$ and measurements $M_{b|y,c} = \delta_{c,b}B_{c|y} + \delta_{b,\varnothing}(\one-B_{c|y})$.

In the same way that Ref.~\cite{Masini2024} reformulated the guessing attacks of Ref.~\cite{Acin2016} through the notions of full and partial-input JM---thereby enabling the discovery of improved strategies---part of the motivation of the present work was to introduce a generalized notion of partial joint-measurability allowing us to reformulate the guessing attack of Ref.~\cite{Chaturvedi2024} within this new framework. We show below that exploiting this generalized framework allows us to find improved strategies for postselected scenarios, which go beyond the simple guessing attacks of Refs.~\cite{Chaturvedi2024} and~\cite{Acin2016}.

\vspace{-6pt}
\subsubsection*{(d) Partial input \& outcome JM} This case represents the intersection of cases \emph{(b)} and \emph{(c)}: Eve focuses only on the conclusive outcomes of the single measurement $y=1$. In a QKD context, this corresponds to a protocol where the raw key is extracted from the $y=1$ measurement and no-click events are discarded during key generation. This setup is exactly that of the one-sided DI QKD protocol based on steering introduced in Ref.~\cite{Branciard2012}.

By Observation~\ref{obs:1}, the bound $\eta\leq 1/2$ from case \emph{(b)} applies here as well.
However, we show that a more refined strategy can be implemented, which improves the bound to $\eta \le k/(2k-1)$, recovering the value $1/2$ in the limit $k\to\infty$. The strategy is a generalization of the one in Ref.~\cite{Acin2016} and Ref.~\cite{Chaturvedi2024}, where Eve performs a weak measurement of $B_1$ that allows her to increase the probability of correctly guessing the conclusive outcome, while Bob's device implements in the case $y\neq 1$ a probabilistic reversal of the disturbance caused by Eve's weak measurement.

First, note that, although Bob's measurement $B_1=\{B_{b|1}\}_{b=1}^k$ is generically a POVM acting on the incoming state $\rho$, we can, through Naimark's extension, represent it as a projective measurement $P=\{P_{b}\}_{b=1}^k$ acting on the joint state $\tilde \rho=\rho\otimes \ketbra{0}$, where $|0\rangle$ is an auxiliary system. We define our simulation on this extended space.

Eve implements a quantum instrument $\ins = \{\ins_c\}_{c=1}^k$ defined by the Kraus operators
\begin{equation}
	K_c = \sqrt{\eta} \, P_c + \sqrt{\frac{1-\eta}{k-1}}\, (\one - P_c)\,,
\end{equation}
where $P_c$ are the projectors of the Naimark extension of $B_1$. The associated POVM element $E_c = K_c^\dagger K_c$ represents a weak measurement of the projective measurement $P$, i.e.,
\begin{equation}
	E_c = (1-\nu)\frac{\one}{k} + \nu P_c\qquad \text{with} \quad \nu \coloneqq \frac{k\eta-1}{k-1}\,.
\end{equation}
This construction interpolates between a completely uninformative measurement ($\nu=0$, $\eta = 1/k$) and a sharp projective measurement ($\nu=1$, $\eta=1$).

Upon obtaining the outcome $c$ and the post-measurement state $\ins_c(\tilde \rho)$, Bob's device performs the measurement $M_{y,c}$. For $y=1$, the device measures $P$ and outputs $b$ if $b=c$; otherwise, it returns the no-click outcome $\varnothing$. Formally, $M_{b|1,c} = \delta_{c,b} P_c + \delta_{b,\varnothing} (\one - P_c)$. This reproduces the correct statistics of $B_{1}^\eta$.
Unlike the random guessing strategy in Ref.~\cite{Chaturvedi2024} (which corresponds to $\nu=0$, $\eta=1/k$), this weak measurement allows Eve to increase her probability of correctly guessing the conclusive outcome and hence increases $\eta$. In the limit $\nu=1$, $\eta=1$, Eve performs the sharp measurement $B_1$, as in the strategy of Ref.~\cite{Acin2016}.

However, the instrument $\ins$ introduces a disturbance to the state $\ins_c(\rho)$, which would typically bias the statistics of the measurements $y\neq 1$. To compensate for this, the device implements a two-step procedure for $y\neq 1$. First, upon learning $c$, it applies a two-valued instrument $\insi_c=\{\insi_{\text{inv}|c},\insi_{\varnothing|c}\}$,  designed to probabilistically invert the effect of $\ins_c$. If the outcome $\text{inv}$ is obtained, then Bob's device has successfully reversed the disturbance caused by $\ins_c$ and recovers the original state $\tilde \rho$; it then performs the honest measurement $B_y$ on this state and outputs the outcome obtained. Otherwise, if the outcome $\varnothing$ is obtained, the disturbance has not been successfully reversed, and the device outputs $\varnothing$. The Kraus operator for $\insi_{\text{inv}|c}$ is:
\begin{equation}
	L_{\text{inv}|c} = \sqrt{\gamma}\left(\sqrt{\frac{1}{\eta}} P_c+\sqrt{\frac{k-1}{1-\eta}} (\one - P_c)\right)\,, 
\end{equation}
where $\gamma > 0$ is a normalization constant. This ensures that $L_{\text{inv}|c} K_c = \sqrt{\gamma} \one$. For $\insi_c$ to be a valid instrument, we require that ${L^\dagger_{\text{inv}|c}} L_{\text{inv}|c} \le \one$. Since $P_c$ and $(\one-P_c)$ are mutually orthogonal projectors, this implies, $\gamma \le \eta$ and $\gamma \le (1-\eta)/(k-1)$.

Applying $\insi_{\text{inv}|c}$ to the post-measurement state $\ins_c(\tilde \rho)$ yields the unnormalized state $(\insi_{\text{inv}|c} \circ \ins_c)(\tilde \rho)= L_{\text{inv}|c} K_c \left(\tilde \rho\right) K_c^\dagger L_{\text{inv}|c}^\dagger = \gamma \tilde \rho$, and hence $\sum_{c=1}^{k} (\insi_{\text{inv}|c} \circ \ins_c)(\tilde \rho) = k\gamma \tilde \rho = k\gamma\left(\rho\otimes\ketbra{0}\right)$. This means that with probability $k\gamma$, the effect of the disturbance from the instrument $\ins$ has been successfully corrected by $\insi_c$, and the effective channel behaves like the identity channel. Bob's device then performs the ideal measurement $B_y$ ($y\neq 1$) on this state. The resulting effective POVMs are thus $\tilde B_{b|y} = k\gamma B_{b|y}$ for $b\neq \varnothing$, and $\tilde B_{\varnothing|y} = \one - \tilde B_{b|y} =(1-k\gamma)\one$. To reproduce the honest statistics $B_{b|y}^\eta$, we require $k\gamma = \eta$. This choice is compatible with the other constraints on $\gamma$ stemming from the positivity of the instrument $\insi_c$ if and only if $\eta \le k/(2k-1)$. This completes the proof of the bound for $\mathcal{G}$-JM.

\subsection{Qubit observables}\label{sec:pair_qubit_observables}
While the bounds that are presented in Table~\ref{tab:cases} are generic bounds that apply to any set of measurements $B_y$, we now specialize further to the case of qubit observables of the form 
\begin{equation}\label{eq:qubit_observables}
	B_y = \vec r_y \cdot \vec \sigma ,
	\qquad y\in \{1,\dots,n\},
\end{equation}
where $\vec r_y\in\mathbb{R}^3$ are unit Bloch vectors, $\|\vec r_y\|=1$, and $\vec\sigma=(X,Y,Z)$ is the vector of Pauli matrices. The corresponding POVM elements are the rank-one projectors
\begin{equation}\label{eq:binary_qubit_observables}
	B_{\pm|y}
	=
	\frac{1}{2}
	\left(\one \pm \vec r_y\cdot\vec\sigma\right) \eqqcolon P_{\pm\vec r_y}\,,
\end{equation}
where we label the outcomes by $b\in\{+,-\}$ for convenience.

\vspace{-6pt}
\subsubsection*{(a) Full JM}
The generic bound $\eta \le 1/n$ from Table~\ref{tab:cases} obviously applies also to qubit observables. It cannot be improved by exploiting the specific structure of the measurements~\eqref{eq:qubit_observables}, as it is known that any set of $n$ distinct such observables are fully jointly measurable if and only if $\eta \le 1/n$~\cite{Skrzypczyk2015}, regardless of their measurement directions. 

\vspace{-6pt}
\subsubsection*{(b) Partial-input JM}
The generic bound $\eta \le 1/2$ from Table~\ref{tab:cases} applies here. To show it is tight, it suffices to consider $n=2$. But for $n=2$, cases \emph{(a)} and \emph{(b)} coincide (Observations~\ref{obs:2} and~\ref{obs:3}), and the bound is then tight since any two distinct qubit observables of the form~\eqref{eq:qubit_observables} saturate the full-JM bound~\cite{Skrzypczyk2015}.

\vspace{-6pt}
\subsubsection*{(c) Partial-outcome JM}
\begin{figure}
	\centering
\includegraphics[width=0.24\textwidth]{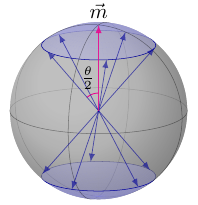}
	\caption{Illustration of the angle $\theta$ defined in Eq.~\eqref{eq:theta_qubit_case_c_general}. It is the angular aperture of the double cone that contains all the measurement axes while having the smallest aperture. This figure shows five measurement axes $\{\pm \vec r_y\}_{y=1}^5$ (blue), and the axis $\vec m$ (red) of the corresponding double cone. 
	}\label{fig:bloch_cap}
\end{figure}
For case \emph{(c)}, where Eve aims to guess only the conclusive outcomes of all measurements, we find that we can go beyond the simple guessing strategies of the previous section to obtain tighter bounds for qubit measurements of the form \eqref{eq:qubit_observables}.
We show that the corresponding lossy measurements $B_y^\eta$ are
$\mathcal{G}$-JM whenever
\begin{equation}\label{eq:bound_qubit_case_c_general}
	\eta \le \frac{1}{1+\sin(\theta/2)}\,,
\end{equation}
where $\theta$ is defined by
\begin{equation}\label{eq:theta_qubit_case_c_general}
	\frac{\theta}{2}
	\coloneqq
	\min_{\|\vec m\|=1}
	\max_{y=1,\dots,n}
	\arccos |\vec m\cdot \vec r_y|.
\end{equation}
For a fixed unit vector $\vec m$, the quantity $\arccos |\vec m\cdot \vec r_y|$ is the smaller angle between $\vec m$ and the measurement axis defined by the two antipodal directions $\{\pm \vec r_y\}$. Thus, $\max_y \arccos |\vec m\cdot \vec r_y|$ is the smallest half-aperture of a double cone with axis $\vec m$ that contains all the measurement axes. Optimizing over $\vec m$, the angle $\theta/2$ is the smallest such half-aperture, or equivalently, $\theta$ is the angular aperture of the double cone that contains all the measurement axes $\{\pm \vec r_y\}_{y=1}^n$ while having the smallest aperture, see Fig.~\ref{fig:bloch_cap}. Hence $\theta$ is close to zero when all
measurements are close to a common axis, and it is large when no single
axis is well aligned with all of them.

In the special case of two measurements, 
\begin{equation}\label{eq:mu_qubit_case_c_general}
	\theta
	=
	\arccos|\vec r_1\cdot\vec r_2|
	\in[0,\pi/2]
\end{equation}
is the angle between the two measurement axes and $\vec m$ is the bisector of the two axes.
For $\theta = \pi/2$, i.e., mutually anticommuting measurements, this yields the bound $\eta \le 2-\sqrt{2} \approx 0.586$. For $\theta \rightarrow 0$, the right-hand side goes to 1 and the bound becomes $\eta \leq 1-\theta/2$. 

More generally, for any finite number $n$ of distinct measurements, $\theta$ is strictly smaller than $\pi$, hence the bound~\eqref{eq:bound_qubit_case_c_general} is strictly larger than $1/2$, and thus always improves over the generic bound from Table~\ref{tab:cases} given by the strategy of Ref.~\cite{Chaturvedi2024}.

We now prove Eq.~\eqref{eq:bound_qubit_case_c_general} by constructing an explicit strategy for Eve. Fix a unit vector $\vec m$ and a parameter $0\le \nu<1$, which will both be fixed later. Eve first applies the two-outcome quantum instrument $\ins=\{\ins_c\}_{c\in\{\pm\}}$ defined by the Hermitian Kraus operators
\begin{equation}\label{eq:Kraus_qubit_case_c_general}
	K_{\pm}
	\coloneqq
	\sqrt{\frac{1+\nu}{2}}\,
	P_{\pm\vec m}
	+
	\sqrt{\frac{1-\nu}{2}}\,
	P_{\mp\vec m},
\end{equation}
The corresponding POVM elements are
\begin{equation}
	E_{\pm} = K_{\pm}^2
	=
	\frac{1}{2}
	\left(\one\pm \,\nu\,\vec m\cdot\vec\sigma\right)\,,
\end{equation}
corresponding to a weak measurement along the direction $\vec m$.
Given Eve's outcome $c$, her guess $g_y(c)$ for the conclusive outcome of Bob's measurement $y$ is $c$ if the Bloch vector $\vec r_y$ has positive overlap with $\vec m$ and $-c$, otherwise. More explicitly,
\begin{equation}\label{eq:guess_function_qubit_case_c_general}
	g_y(c) =
	c\,\operatorname{sgn}(\vec m\cdot\vec r_y).
\end{equation}

In order to ensure that Eve's guess is always correct whenever Bob produces a conclusive outcome, Bob's measurements, conditioned on Eve's outcome $c$, must satisfy
\begin{equation}
	M_{b|y,c}
	=
	\delta_{b,g_y(c)}\,M_{\star|y,c},
	\qquad b\in\{+,-\}.
\end{equation}
Thus, the conditional measurement performed by Bob's device is effectively a two-outcome measurement with outcomes $\{\star,\varnothing\}$: if the outcome $\star$ is obtained, Bob outputs $b=g_y(c)$, while if the outcome $\varnothing$ is obtained, Bob outputs the inconclusive outcome $\varnothing$.
We take
\begin{equation}\label{eq:Bob_measurements_qubit_case_c_general}
	M_{\star|y,c}
	=
	\eta\,K_c^{-1} B_{g_y(c)|y} K_c^{-1},\,
	M_{\varnothing|y,c}
	=
	\one-M_{\star|y,c},
\end{equation}
where
\begin{equation}\label{eq:Krausinv_qubit_case_c_general}
	K_{\pm}^{-1}
	\coloneqq
	\sqrt{\frac{2}{1+\nu}}\,
	P_{\pm\vec m}
	+
	\sqrt{\frac{2}{1-\nu}}\,
	P_{\mp\vec m}
\end{equation}
are the inverses of the Kraus operators $K_{\pm}$. 
This construction reproduces the desired lossy statistics. Indeed, for $b\in\{+,-\}$, the only value of $c$ contributing to $M_{b|y,c}$ is 
$c= b\,\operatorname{sgn}(\vec m\cdot\vec r_y)=\tilde c $
and therefore $\sum_c \Tr[K_c \rho K_c^\dagger M_{b|y,c}]= \Tr[ K_{\tilde c} \rho K_{\tilde c}^\dagger M_{\star|y,\tilde c}] = \eta \Tr[\rho B_{b|y}]$. 
The inconclusive statistics are then automatically reproduced by normalization.

It remains to ensure that Eq.~\eqref{eq:Bob_measurements_qubit_case_c_general} defines valid POVM elements. This reduces to the constraint that $M_{\star|y,c}\le \one$, which can be satisfied by choosing $\eta$ sufficiently small. Since the operators $B_{b|y}$ are rank-one projectors, the operators $M_{\star|y,c}$ are also rank-one, and hence their unique non-zero eigenvalue is given by
\begin{align}
\Tr[M_{\star|y,c}]  & =  \eta \left(\frac{2}{1+\nu} \Tr[P_{c \vec m} P_{g_{y}(c)\vec r_y}]\right.\nonumber \\
&\quad + \left.\frac{2}{1-\nu} \Tr[P_{-c \vec m} P_{g_{y}(c)\vec r_y}]\right)\nonumber \\
&  = \eta\left(\frac{1+|\vec m \cdot \vec r_y|}{1+\nu} + \frac{1-|\vec m \cdot \vec r_y|}{1-\nu} \right)\nonumber\\
& = \eta\times\frac{2(1-\nu |\vec m \cdot \vec r_y|)}{1-\nu^2}.
\end{align}
Thus, the condition $M_{\star|y,c}\le\one$ is equivalent to
\begin{equation}\label{eq:constraint_qubit_case_c_general_y}
	\eta
	\le
	\frac{1-\nu^2}{2(1-\nu |\vec m \cdot \vec r_y|)},
\end{equation}
for all $y$. Defining
\begin{equation}
	\mu(\vec m)
	\coloneqq
	\min_{y=1,\dots,n}
	|\vec m\cdot\vec r_y|,
\end{equation}
we obtain the sufficient condition
\begin{equation}\label{eq:constraint_qubit_case_c_general_m}
	\eta
	\le
	\frac{1-\nu^2}
	{2\left(1-\nu\,\mu(\vec m)\right)}.
\end{equation}
For fixed $\vec m$, the right-hand side of Eq.~\eqref{eq:constraint_qubit_case_c_general_m} is maximized over $0\le\nu<1$ by choosing
\begin{equation}\label{eq:nu_qubit_case_c_general}
	\nu
	=
	\frac{1-\sqrt{1-\mu(\vec m)^2}}
	{\mu(\vec m)}
	=
	\frac{\mu(\vec m)}
	{1+\sqrt{1-\mu(\vec m)^2}}.
\end{equation}
This gives
\begin{equation}
	\eta
	\le
	\frac{1}
	{1+\sqrt{1-\mu(\vec m)^2}}.
\end{equation}
Finally, the best measurement direction $\vec m$ is the one maximizing $\mu(\vec m)$, corresponding to $\mu = \cos(\theta/2)$ for $\theta/2$ defined in Eq.~\eqref{eq:theta_qubit_case_c_general}, which yields the bound~\eqref{eq:bound_qubit_case_c_general}.

In the case $n=2$, we found numerically, by solving the $\mathcal{G}$-JM SDP for a large number of angles $\theta \in [0,\pi/2]$ that the bound~\eqref{eq:bound_qubit_case_c_general} is tight up to numerical precision. By Observation~\ref{obs:2}, this implies also that the general bound \eqref{eq:bound_qubit_case_c_general} is tight for $n$ observables when their Bloch axes are all contained in a double cone of angular aperture $\theta \leq \pi/2$.

For values of $\theta\in \, ]\pi/2,\pi]$, we tried to saturate numerically the bound~\eqref{eq:bound_qubit_case_c_general} using a small set of $n$ observables parametrized as
\begin{equation}\label{eq:qubit_meas_case_c}
\vec r_y = (\sin(\theta/2) \cos(\phi_y), \sin(\theta/2) \sin(\phi_y), \cos(\theta/2))	\,,
\end{equation}
with $\phi_y =  2\pi (y-1)/n$.
These measurement axes are placed symmetrically around the double cone with axis $\vec m = \hat z$ and angular aperture $\theta$.

As already mentioned, with $n=2$ such measurements, the bound~\eqref{eq:bound_qubit_case_c_general} is tight for $\theta \le \pi/2\approx 1.57$. With $n=3$ observables of the above form, we can numerically saturate~\eqref{eq:bound_qubit_case_c_general} for $\theta \le 2\arccos(1/\sqrt{3})\approx 1.91$. Note that the observables~\eqref{eq:qubit_meas_case_c} with such a value of $\theta$ correspond, up to a change of basis, to the mutually anticommuting $\{X,Y,Z\}$ observables. We also observe numerically that as we increase the number of measurements $n$, the bound~\eqref{eq:bound_qubit_case_c_general} appears to be tight for a larger range of $\theta$, up to $\theta \lesssim 1.952$ with $n=9$ observables. Whether the bound~\eqref{eq:bound_qubit_case_c_general} is tight for all angles $\theta$ remains an open question.

\vspace{-6pt}
\subsubsection*{(d) Partial input \& outcome JM} \label{sec:qubit_case_d}
For case \emph{(d)}, where Eve is required to guess the conclusive outcomes of the first measurement $y=1$ only, the generic bound provided in Table~\ref{tab:cases} reduces to $\eta \le 2/3$ for the case of the qubit observables \eqref{eq:qubit_observables} since $k=2$. However, we can improve this bound using the specific structure of the qubit measurements. 

In the case of $n=2$ measurements separated by an angle $\theta \in [0,\pi/2]$ on the Bloch sphere, we find that the bound can be improved to 
\begin{equation}\label{eq:bound_qubit_case_d_n2}
	\eta \le \frac{2}{2+\sin(\theta)}.
\end{equation}
For $\theta = \pi/2$, this recovers the generic bound $\eta \le 2/3$, while for $\theta < \pi/2$, this is a strictly better bound, with the right-hand side going to 1 for $\theta \rightarrow 0$. We find numerically by solving the $\mathcal{G}$-JM SDP for a large number of values of $\theta \in [0,\pi/2]$ that, up to machine precision, this bound is tight. 

In the case of an arbitrary number $n$ of qubit observables, we find the bound
\begin{equation}\label{eq:bound_qubit_case_d_n}
	\eta \le \frac{1+\sin(\theta)}{1+2\sin(\theta)},
\end{equation}
where
\begin{equation}\label{eq:theta_qubit_case_d_general}
	\theta \coloneqq \max_{y=2,\ldots,n} \arccos |\vec r_1\cdot\vec r_y|,
\end{equation}
with $\theta\in[0,\pi/2]$. Here $\theta$ is the largest angle between
the first measurement axis $\{\pm\vec r_1\}$ and the remaining
measurement axes $\{\pm\vec r_y\}_{y=2}^{n}$, see Fig.~\ref{fig:bloch_cap2}. For $\theta = \pi/2$,  this recovers again the generic bound $\eta \le 2/3$, but improves it for $\theta < \pi/2$. We find strong numerical evidence that this bound is tight using only three measurements of the form $\vec r_1 = \hat z$, $\vec r_2 = \sin \theta \hat x + \cos \theta \hat z$ and $\vec r_3 = -\sin \theta \hat x + \cos \theta \hat z$ by solving the corresponding SDP for a large number of values of $\theta \in [0,\pi/2]$. Hence, using Observation~\ref{obs:2}, we conclude that the above bound is tight for any number of measurements $n$ as long as all measurement axes are contained in a double cone of angular aperture $\theta$ around the axis defined by the first measurement.
\begin{figure}
	\centering
\includegraphics[width=0.24\textwidth]{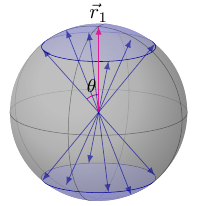}
	\caption{Illustration of the angle $\theta$ defined in Eq.~\eqref{eq:theta_qubit_case_d_general}. It is the largest angle between
	the first measurement axis $\{\pm\vec r_1\}$ (red) and the remaining measurement axes $\{\pm\vec r_y\}_{y=2}^{n}$ (blue), shown here for $n=7$. Equivalently, $2\theta$ is the angular aperture of the double cone with axis $\vec r_1$ that contains all the measurement axes while having the smallest aperture. 
	}\label{fig:bloch_cap2}
\end{figure}

To prove these bounds, we construct explicit strategies, which are initially identical to the one for case \emph{(c)}: Eve applies a two-outcome instrument \eqref{eq:Kraus_qubit_case_c_general} corresponding to a weak measurement along a direction $\vec m$ with strength $\nu$. Her guess for the conclusive outcome of the first measurement $y=1$ is chosen as before as $g_1(c) = c\,\operatorname{sgn}(\vec m\cdot\vec r_1)$ (However, note that the guess functions $g_y(c)$ for $y\neq 1$ are not relevant here since Eve does not need to guess the outcomes of these measurements).

Bob's conditional measurement for $y=1$ is then, as before,
\begin{equation}\label{appeq:Bob_measurement_y1_case_d_general}
	M_{b|1,c}
	=
	\delta_{b,g_1(c)}M_{\star|1,c},\quad
	M_{\star|1,c}
	=
	\eta K_c^{-1}B_{g_1(c)|1}K_c^{-1},
\end{equation}
and
\begin{equation}
	M_{\varnothing|1,c}
	=
	\one-M_{\star|1,c}.
\end{equation}
This guarantees that, whenever Bob outputs a conclusive outcome for
$y=1$, Eve's guess is correct, and furthermore that the correct lossy statistics are reproduced for $y=1$. As before, this forms a valid POVM if
\begin{equation}\label{eq:constr1}
	\eta \leq \frac{1-\nu^2}{2(1-\nu |\vec{m}\cdot \vec{r}_1|)} \eqqcolon F_{\nu,t_1}\,,
\end{equation}
where $F_{\nu,t} \coloneqq (1-\nu^2)/[2(1-\nu t)]$ and we introduce the notation
\begin{equation}
	t_y \coloneqq |\vec m \cdot \vec r_y|\,,
\end{equation}
that will be used repeatedly in the following.

For the other measurements, $y\neq1$, since Eve does not need to guess the outcomes, we have more freedom in choosing the operators $M_{b|y,c}$. A natural and simple choice is to make them proportional to $\eta  K^{-1}_{c} B_{b|y} K^{-1}_{c}$ for $b \neq \varnothing$, i.e.,
\begin{align}
	M_{b|y,c}=
	q_{b,c|y}\,\eta K_c^{-1}B_{b|y}K_c^{-1},
	\qquad b\in\{+,-\},
\end{align}
and 
\begin{equation}
	M_{\varnothing|y,c}=\one-M_{+|y,c}-M_{-|y,c},
\end{equation}
where $q_{b,c|y}\geq 0$. We further impose $\sum_{c} q_{b,c|y}=1$ for every fixed $b$ and $y$, so that $\sum_{c} \Tr[K_c \rho K_c^\dagger M_{b|y,c}] = \eta (\sum_c q_{b,c|y}) \Tr[\rho B_{b|y}] = \eta \Tr[\rho B_{b|y}]$ and hence the correct lossy statistics are reproduced for $y\neq 1$. The above operators form valid POVM elements as long as $M_{+|y,c}+M_{-|y,c}\leq \one$ for all $y\neq 1$ and $c$. We have
\begin{align}\label{eq:kck}
	M_{+|y,c}+M_{-|y,c}
	&=
	\eta K_c^{-1} \left(\sum_{b=\pm} q_{b,c|y} B_{b|y}\right) K_c^{-1} \nonumber\\
	&=
	\eta K_c^{-1} C_{y,c} K_c^{-1}\,,
\end{align}
where we defined $C_{y,c} \coloneqq \sum_{b} q_{b,c|y} B_{b|y}$.
To simplify this expression, we further choose $q_{b,c|y}$ such that $\sum_{b} q_{b,c|y} =1$ for every fixed $c$ and $y$. Thus, for fixed $y$, $q_{b,c|y}$ is a doubly stochastic matrix, which is therefore defined by a single parameter $\gamma_y\in [0,1]$, allowing us to write in full generality
\begin{equation}
	q_{b,c|y} = \frac{1}{2} \left(1 + b \, g_y(c)\gamma _y\right)\,,
\end{equation}
with $g_y(c)=c\,\operatorname{sgn}(\vec m\cdot\vec r_y)$.
We then have
\begin{equation}
	C_{y,c} = \sum_{b=\pm} q_{b,c|y} B_{b|y} = \frac{1}{2}\left(\one + \gamma_y g_y(c) \vec r_y \cdot \vec \sigma\right),
\end{equation}
and the Bloch vector of $C_{y,c}$ is aligned with that of $B_{g_y(c)|y}$ and has length $\gamma_y$. 
One can then  easily check that the trace of the operator~\eqref{eq:kck} is $2 \eta A_y/(1-\nu^2)$, where
\begin{equation}
A_y = 1 - \nu \gamma_y |\vec m \cdot \vec r_y| = 1 - \nu \gamma_y t_y,
\end{equation}
while its determinant is
\begin{equation}
	\begin{aligned}
	\det[\eta K_c^{-1} C_{y,c} K_c^{-1}] &= \eta^2 \det[(K_c^{-1})^2] \det[C_{y,c}]\,\\
	&=  \eta^2 \,\frac{4}{1-\nu^2} \, \frac{1-\gamma_y^2}{4}\\
	&= \eta^2 \frac{1-\gamma_y^2}{1-\nu^2}\,,
	\end{aligned}
\end{equation}
where we used $\det[\one+ \vec u \cdot \vec \sigma] = 1- \vec u \cdot \vec u$ for any vector $\vec u$. This implies that the largest eigenvalue of $M_{+|y,c}+M_{-|y,c}$ is
\begin{equation}
	\lambda_{\max}^{(y)} = \eta \frac{A_y+\sqrt{A_y^2-(1-\nu^2)(1-\gamma_y^2)}}{1-\nu^2}.
\end{equation}
Thus, the constraint from measurement $y\neq 1$ is
\begin{equation}
	\eta \leq \frac{1-\nu^2}{A_y+\sqrt{A_y^2-(1-\nu^2)(1-\gamma_y^2)}}.
\end{equation}
For fixed $\nu$ and $t_y$, we want to maximize the right-hand side over the choice of $\gamma_y$, so that
\begin{equation}\label{eq:constr2}
\eta\leq G_{\nu,t_y} \coloneqq \max_{0\le \gamma_y \le 1} \frac{1-\nu^2}{A_y+\sqrt{A_y^2-(1-\nu^2)(1-\gamma_y^2)}}\,.
\end{equation}
Setting the derivative of the denominator of the right-hand side to zero yields the stationary point $\gamma_y = \nu t_y / (1-\nu \sqrt{1-t_y^2}\,) \ge 0$. 
Imposing further that $\gamma_y \le 1$, we obtain
\begin{equation}
	G_{\nu,t_y}
	=
	\begin{cases}
	1-\nu\sqrt{1-t_y^2},
	&
	\text{if } \nu\left(t_y+\sqrt{1-t_y^2}\right)\le1,
	\\[1.2ex]
	\dfrac{1-\nu^2}{2(1-\nu t_y)},
	&
	\text{if } \nu\left(t_y+\sqrt{1-t_y^2}\right)\ge1.
	\end{cases}
\end{equation}
Combining the constraint \eqref{eq:constr1} from $y=1$ and \eqref{eq:constr2} from $y\neq 1$, we therefore obtain the following general sufficient condition for $\mathcal{G}$-JM	
\begin{equation}\label{appeq:general_finite_n_case_d_bound}
	\eta
	\le
	\max_{\|\vec m\|=1,\,0\le\nu<1}
	\min\left\{
		F_{\nu,t_1	},
		\min_{y\neq1}G_{\nu,t_y}
	\right\}.
\end{equation}

Let us now specialize to the case of $n=2$ measurements that are separated by an angle $\theta\in[0,\pi/2]$. Without loss of generality, take
\begin{equation}
	\vec r_1=\hat z,
	\qquad
	\vec r_2=\cos(\theta)\,\hat z+\sin(\theta)\,\hat x.
\end{equation}
We choose $\vec m$ in the plane spanned by $\vec r_1$ and $\vec r_2$,
making an angle $x\in[0,\theta]$ with $\vec r_1$:
\begin{equation}
	\vec m
	=
	\cos(x)\,\hat z+\sin(x)\,\hat x.
\end{equation}
Then, optimizing over the angle $x$ and $\nu$ gives, as shown in Appendix~\ref{app:qubit_case_d}, the bound~\eqref{eq:bound_qubit_case_d_n2}.

Finally, to derive the simple closed-form bound~\eqref{eq:bound_qubit_case_d_n} for an arbitrary number of measurements in the double cone of angular aperture $2\theta$ around the first measurement axis (see Fig.~\ref{fig:bloch_cap2}), we choose
\begin{equation}
	\vec m=\vec r_1.
\end{equation}
Optimizing over $\nu$ then gives the bound~\eqref{eq:bound_qubit_case_d_n}, as shown in Appendix~\ref{app:qubit_case_d}.

Note that for $n=2$, the bound \eqref{eq:bound_qubit_case_d_n} is more conservative than the bound \eqref{eq:bound_qubit_case_d_n2}, since
\begin{equation}
	\frac{2}{2+\sin(\theta)}
	-
	\frac{1+\sin(\theta)}{1+2\sin(\theta)}
	=
	\frac{\sin(\theta)(1-\sin(\theta))}
	{(2+\sin(\theta))(1+2\sin(\theta))}
	>0.
\end{equation}
This difference originates from the fact that the $n=2$ construction tunes
the weak-measurement direction $\vec m$ to the specific pair of
measurements, whereas the one for arbitrary $n$ fixes $\vec m=\vec r_1$ and
works uniformly for all measurements whose angle with $\vec r_1$ is at most $\theta$.

\subsection{Impact of finite visibility}
\begin{figure}
    \centering
    \begin{subfigure}{0.3\textwidth}
        \centering
        \includegraphics[width=\textwidth]{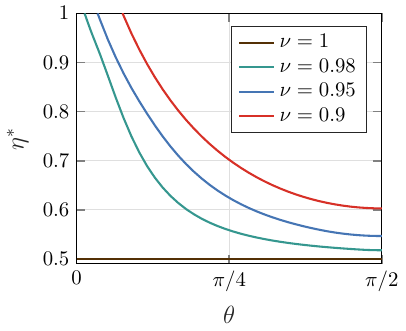}
         \caption*{Cases~\emph{(a)} and~\emph{(b)}}
        \end{subfigure}
		\hfill
        \begin{subfigure}{0.3\textwidth}
        \centering
        \includegraphics[width=\textwidth]{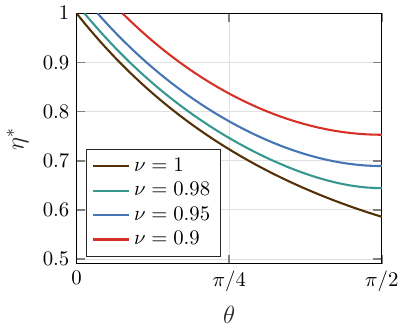}
        \caption*{Case~\emph{(c)}}\label{fig:case_c}
        \end{subfigure}
		\hfill
        \begin{subfigure}{0.3\textwidth}
        \centering
        \includegraphics[width=\textwidth]{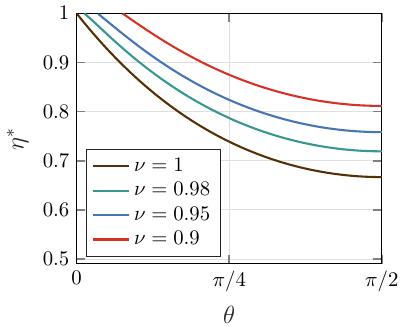}
        \caption*{Case~\emph{(d)}}\label{fig:case_d}
        \end{subfigure}
    \caption{The maximum detection efficiency $\eta^* $ for which the lossy measurements $\{Z, \cos(\theta) Z+ \sin(\theta) X\}$ are $\mathcal{G}$-JM at different values of visibility $\nu$.}
    \label{fig:eta_vs_angle_vis}
\end{figure}
So far, we have considered detection loss as the only source of imperfection. We now also take into account the effect of finite visibility. We can model the combined effect of detection loss and finite visibility by modifying the ideal POVM elements $B_{b|y}$ into effective POVM elements $B^{\eta,\nu}_{b|y}$ given by
\begin{equation}
	B^{\eta,\nu}_{b|y} = 
	\begin{cases}
		\eta \, \nu \, B_{b|y} + \eta \, (1-\nu) \, t_{b|y} \one\,, \quad & b \neq \varnothing\,,\\
		(1-\eta) \one\,,  &b = \varnothing\,,
	\end{cases}
\end{equation}
where $\eta$ is the detection efficiency, $\nu$ the visibility, and $t_{b|y} := \Tr(B_{b|y})/\Tr(\one)$. The maximum detection efficiency $\eta^*$ for which the measurements $B^{\eta,\nu}_y$ are $\mathcal{G}$-JM for the four cases considered above can be computed by solving the SDP in Def.~\ref{defn:partial-joint-measurability2} for different values of $\theta$ and $\nu$. We plot the results in Fig.~\ref{fig:eta_vs_angle_vis} for two measurements of the form $\{Z, \cos(\theta) Z+ \sin(\theta) X\}$ separated by an angle $\theta$ in the Bloch sphere. Note that, as explained previously, cases \emph{(a)} and \emph{(b)} are equivalent. As expected, the values of the threshold efficiency $\eta^*$ for cases \emph{(c)} and \emph{(d)} are higher since postselection only makes it easier to guess the measurement outcomes.

\subsection{Application to the security of the one-sided DI QKD protocol of ~\cite{Branciard2012}}\label{sec:Branciard}

As an application of our results, we show a counterexample to the security proof of the one-sided DIQKD protocol proposed in~\cite{Branciard2012}, which is based on quantum steering. In this protocol, only one party is assumed to have a characterized measurement apparatus.

In this setup, the trusted party, Alice\footnote{Note that we have swapped the roles of Alice and Bob with respect to Ref.~\cite{Branciard2012}.}, performs projective measurements $A_1$ and $A_2$ corresponding to the qubit observables $Z$ and $X$. The untrusted party, Bob, chooses between two binary measurements $B_1$ and $B_2$, which in the honest implementation are also $Z$ and $X$ basis measurements. The secret key is extracted from the conclusive (click) outcomes of the $A_1$ and $B_1$ settings, while no-click events are discarded. (The no-click events are still retained for parameter estimation).
 Since Alice's device is trusted, her non-detection events cannot be manipulated by an adversary; however, Eve is free to exploit the non-detection events on Bob's side to gain information about the conclusive outcomes.

As established in Sec.~\ref{sec:Eve_attack}, if Bob's measurements are $\mathcal{G}$-jointly measurable for $\mathcal{G}_1=\{+,-\}$ and $\mathcal{G}_2=\emptyset$, then Eve can obtain a perfect copy of Bob's conclusive outcomes for the $B_1$ measurement. In such a case, no secret key can be extracted in the protocol of~\cite{Branciard2012}.
As shown in Sec.~\ref{sec:pair_qubit_observables}, the qubit observables $\{Z,X\}$ are $\mathcal{G}$-JM whenever the detection efficiency satisfies $\eta \le 2/3$. 

This result contradicts the security analysis in Ref.~\cite{Branciard2012}, where it is claimed that the protocol remains secure for $\eta > 0.659$. We identify the source of this discrepancy in the treatment of postselection within their security proof. 

To prove security, the authors of Ref.~\cite{Branciard2012} seek to bound the smooth min-entropy $H^{\epsilon}_{\text{min}}(\mathsf{A}^{\text{ps}}|\mathsf{E}')$ of Alice's postselected bit-string $\mathsf{A}^{\text{ps}}$ conditioned on Eve's side information $\mathsf{E}'$. The postselected string $\mathsf{A}^{\text{ps}}$ consists of the outcomes of rounds where both Alice and Bob performed the first measurement ($A_1, B_1$) and Bob obtained a conclusive outcome ($b \neq \varnothing$). The authors utilize the following inequality:
\begin{equation}
	H^{\epsilon}_{\text{min}}(\mathsf{A}^{\text{ps}}|\mathsf{E}') \ge H^{\epsilon}_{\text{min}}(\mathsf{A}|\mathsf{E}') - (N-n)\,,
\end{equation}
where $\mathsf{A}$ is Alice's full string before postselection, $N$ is the total number of rounds, and $n$ is the number of postselected rounds. Here, $\mathsf{E}'=(\mathsf{E},\mathsf{C})$ represents the total information available to Eve, consisting of her initial side information $\mathsf{E}$ and the public announcement $\mathsf{C}$ of which rounds were discarded on Bob's side. 

The flaw in this analysis arises when bounding $H^{\epsilon}_{\text{min}}(\mathsf{A}|\mathsf{E}')$. The authors apply a generalized uncertainty relation~\cite{Tomamichel2011}, but they condition the entropy on the information $\mathsf{E}$ that Eve possesses \emph{before} the postselection process. This would not be a concern if conditioning the entropy of Alice's full string $\mathsf{A}$ (before postselection) on $\mathsf{E}$ (Eve's information before postselection) or $\mathsf{E}'$ (Eve's information after postselection) were equivalent, i.e., if $H^{\epsilon}_{\text{min}}(\mathsf{A}|\mathsf{E}) = H^{\epsilon}_{\text{min}}(\mathsf{A}|\mathsf{E}')$. This might seem plausible at first glance: since the postselection process only involves Bob's measurement outcomes, the public announcement $\mathsf{C}$ of which rounds gave conclusive or inconclusive outcomes on Bob's side does not seem, intuitively, to help Eve guess Alice's full string. However, such intuition is not necessarily valid. Indeed, we show in Appendix~\ref{app:postselection} an explicit example where $H^{\epsilon}_{\text{min}}(\mathsf{A}|\mathsf{E}') < H^{\epsilon}_{\text{min}}(\mathsf{A}|\mathsf{E})$ following the attacks for Eve described in cases \emph{(c)} and \emph{(d)} of Sec.~\ref{sec:pair_qubit_observables}.

This highlights a broader point regarding the role of no-click events in DIQKD. For instance, in Ref.~\cite{Chaturvedi2024}, where the authors show that an adversary can guess every conclusive outcome whenever $\eta \le 1/k$, they claim that ``in QKD protocols this is enough to know all the generated key''. This is obviously true if the key is generated solely from conclusive outcomes. However, if the key generation process incorporates no-click events, a secure positive key could still be extracted. For instance, in the protocol of Ref.~\cite{Branciard2012}, an adversary can perfectly guess the conclusive outcomes of Bob whenever $\eta \le 2/3$ if no-click events are discarded. However, it follows from Refs.~\cite{Masini2024one,LeRoy25,Sekatski2025} that such steering-based protocols can extract a secret key for $\eta > 0.5$ provided that the raw key incorporates the no-click events. We provide a simple example in App.~\ref{app:postselection} illustrating how including non-detection events can preserve secrecy even when conclusive outcomes are perfectly known to the adversary.

\section{Discussion} \label{sec:discussion}

We introduced a generalized notion of partial joint-measurability ($\mathcal{G}$-JM) that extends the framework of Ref.~\cite{Masini2024} by allowing a fine-grained specification, for each measurement setting, of which subset of outcomes must be classically determined. This generalization naturally captures, e.g., the postselection of data that arises in certain implementations of (semi-)device-independent quantum cryptographic protocols, where no-click events are discarded during key generation.

We established an alternative, equivalent formulation of $\mathcal{G}$-JM in terms of partial parent (PP) POVMs, which admits a direct SDP characterization. This makes the question of whether a given set of measurements is $\mathcal{G}$-JM efficiently decidable. We further proved that $\mathcal{G}$-JM has a clean operational interpretation in the adversarial setting: an adversary Eve restricted to classical side information (i.e., without quantum memory) can perfectly predict the outcomes of an untrusted measurement device within the subset $\mathcal{G}_y$ for all input states if and only if the measurements are $\mathcal{G}$-JM. This equivalence generalizes the well-known connection between full joint-measurability and the absence of randomness.

The notion of $\mathcal{G}$-joint-measurability provides a practical and efficiently computable criterion for determining when an untrusted measurement device is entirely useless for DI or semi-DI applications. It complements the existing toolkit of attacks based on local hidden variable models, convex decompositions, and SDP bounds on key rates. Because it depends only on the measurements performed by a single untrusted node—independently of the state shared between the parties—it is particularly valuable in asymmetric scenarios such as one-sided DI protocols, where only one party is assumed to be trusted. 

Several questions remain open. First, extending our results concerning the qubit observables to higher-dimensional measurements is an interesting open problem. Second, while we have focused on the case of a single untrusted measurement device (Bob's side), it would be valuable to develop a corresponding theory for scenarios where both parties are untrusted. Third, $\mathcal{G}$-JM provides a simple, measurement-only attack. A natural next step is to combine it with complementary attack strategies—such as convex combination attacks~\cite{Lukanowski2023} or attacks exploiting the state shared between Alice and Bob -- to obtain tighter upper bounds on key rates for specific protocols. Finally, the connection between $\mathcal{G}$-JM and other resource theories of measurement incompatibility~\cite{Uola2019,Buscemi2020} deserves further investigation.

\section{Acknowledgements} 
We acknowledge funding from the QuantERA II Programme that has received funding from the European Union's Horizon 2020 research and innovation programme under Grant Agreement No 101017733 and the F.R.S-FNRS Pint-Multi programme under Grant Agreement R.8014.21, from the European Union's Horizon Europe research and innovation programme under the project ``Quantum Security Networks Partnership'' (QSNP, grant agreement No 101114043), from the F.R.S-FNRS through the PDR T.0171.22, from the FWO and F.R.S.-FNRS under the Excellence of Science (EOS) programme project 40007526, from the FWO through the BeQuNet SBO project S008323N, from the Belgian Federal Science Policy through the contract RT/22/BE-QCI and the EU ``BE-QCI'' program.
    
S.P. is a Research Director of the Fonds de la Recherche Scientifique - FNRS.
E.P.L. acknowledges support from the Fonds de la Recherche Scientifique - FNRS through a FRIA grant.

Funded by the European Union. Views and opinions expressed are however those of the authors only and do not necessarily reflect those of the European Union. The European Union cannot be held responsible for them.

\bibliography{references.bib}

\appendix

\onecolumngrid

\section{Equivalence between the two formulations of partial joint-measurability}\label{app:equivalence_PJM}

Here, we show that the formulation of partial joint-measurability in terms of effective POVMs ${E}_{c,b|y}$ is equivalent to Def.~\ref{defn:partial-joint-measurability}. Given the instrument $\ins$ and the measurement $M_{y,c}$ in Def.~\ref{defn:partial-joint-measurability}, 
the effective measurement operators ${E}_{c,b|y}$ are defined as
\begin{equation}
	E_{c,b|y} \coloneqq \ins_{c}^\dagger ({M}_{b|y,c})\,,
\end{equation}
which, as shown in Sec.~\ref{sec:sdp_reformulation}, satisfy the conditions
\begin{align}
	&\sum_{b} E_{c,b|y}=E_c \quad \forall y \in  [n]&\text{(no-signaling)}\,, \label{appeq:no-signaling} \\
	&B_{b|y} = \sum_c E_{c,b|y}&\text{(consistency)}\,, \label{appeq:consistency}\\
	&E_{c,b|y} = p(b|y,c)\, E_{c,\star|y} \quad \forall b \in \mathcal{G}_y &\text{(partial JM)}\,,\label{appeq:partial-JM}
\end{align}
for some positive operator $E_{c,\star|y} := E_c - \sum_{b \in \overline{\mathcal{G}}_y} E_{c,b|y}$
independent of $b$, and where $p(b|y,c)$ is a conditional probability distribution satisfying $p(b|y,c)\ge 0$ and $\sum_{b \in \mathcal{G}_y} p(b|y,c) = 1$.

Conversely, for any set of POVMs $E_{y}=\{{E}_{c,b|y}\}_{b,c}$ that satisfies the conditions~\eqref{appeq:no-signaling}--\eqref{appeq:partial-JM}, one can always construct a quantum instrument $\ins=\{\ins_c\}_c$ and measurements $M_{y,c}=\{M_{b|y,c}\}_{b}$ satisfying the conditions of Def.~\ref{defn:partial-joint-measurability}.

To see this, decompose the Hilbert space into a direct sum of the support and kernel of $E_c$, i.e., 
$\mathcal{H} = \mathrm{supp}(E_c)\oplus \ker(E_c)$, and let 
$\Pi_{E_c}$ and $\Pi^\perp_{E_c}$ denote the projectors onto 
$\mathrm{supp}(E_c)$ and $\ker(E_c)$, respectively. 
Define the instrument and POVMs as 
\begin{equation}\label{appeq:real_attack_instrument}
		\ins_c(\cdot)\coloneqq \sqrt{E_c}\,(\,\cdot\,) \sqrt{E_c} \quad \text{and} \quad
		M_{b|y,c} \coloneqq E_c^{-1/2} \,  {E}_{c,b|y} \, E_c^{-1/2} \,,
\end{equation}
where $E_c^{-1/2}$ denotes the Moore-Penrose pseudo-inverse of $\sqrt{E_c}$. The operators $M_{b|y,c}$ are positive and satisfy $\sum_{b} M_{b|y,c} = E_c^{-1/2} \,  \sum_{b} {E}_{c,b|y} \, E_c^{-1/2} = \Pi_{E_c}$. They form valid POVMs on the subspace $\mathrm{supp}(E_c)$ since the states $\ins_c(\rho)$ lie entirely within this support, implying that  $\Pi_{E_c}$ acts as the identity operator on the image of $\ins_c$.

Moreover, it is straightforward to verify that conditions~\eqref{eq:right_prob} and~\eqref{eq:partial-JM-constraint} are satisfied. Indeed, we have for all $y \in  [n]$ and $b \in [k_y]$,
\begin{equation}
	\sum_c \Tr[\ins_{c}(\rho) {M}_{b|y,c}] = \sum_c \Tr[\rho \Pi_{E_c} E_{c,b|y} \Pi_{E_c}] = \sum_c \Tr[\rho E_{c,b|y}] = \Tr[\rho B_{b|y}], 
\end{equation}
where we used that each $E_{c,b|y}$ lies in the support of $E_c$\,, and 
\begin{equation}
	M_{b|y,c} = p(b|y,c) M_{\star|y,c} \quad \forall b \in \mathcal{G}_y \,,
\end{equation}
where $M_{\star|y,c} := E_c^{-1/2} E_{c,\star|y} E_c^{-1/2}$ is a positive operator satisfying $M_{\star|y,c} = \Pi_{E_c} - \sum_{b\in \overline{\mathcal{G}}_y} M_{b|y,c}$ which matches the condition in Def.~\ref{defn:partial-joint-measurability} (with $\Pi_{E_c}$ playing the role of the identity).

\section{Partial joint-measurability with deterministic response functions}\label{app:proof:deterministic_response}
Here we show that if a set of measurements $\{B_y\}_{y\in [n]}$ is $\mathcal{G}$-JM, then there exist POVM elements $E_{\bar{\beta},b|y}$, where $\bar{\beta} = (\beta_1,\ldots,\beta_n)$ with $\beta_y \in \mathcal{G}_y$ (and $\beta_y=\varnothing$ if $\mathcal{G}_y=\varnothing$), such that
\begin{align}
	&\sum_{b} E_{\bar{\beta},b|y} = E_{\bar{\beta}} \quad \forall y \in  [n]\,, &\text{(no-signaling)}\,, \label{appeq:no-signalingSDP}\\
	&B_{b|y} = \sum_{\bar{\beta}} E_{\bar{\beta},b|y}\,,&\text{(consistency)}\,, \label{appeq:consistencySDP} \\
	&E_{\bar{\beta},b|y} = \delta_{b,\beta_y}\,E_{\bar{\beta},\star|y} \quad \forall b \in \mathcal{G}_y\,,&\text{(partial JM)} \label{appeq:partial-JM_SDP}\,, 
\end{align}
for some positive operator $E_{\bar{\beta},\star|y} := E_{\bar{\beta}} - \sum_{b \in \overline{\mathcal{G}}_y} E_{\bar{\beta},b|y}$.
\begin{proof}
	Since the set of measurements $\{B_y\}_{y \in  [n]}$ is $\mathcal{G}$-JM, there exist POVM elements $E'_{c,b|y}$ and response functions $p(b|y,c)$ satisfying conditions~\eqref{appeq:no-signaling}--\eqref{appeq:partial-JM}.
Define the POVM elements $E_{\bar{\beta},b|y}$ as
\begin{align}
	E_{\bar{\beta},b|y} &\coloneqq \sum_c p(\bar{\beta}|c)\,E'_{c,b|y} \quad \forall b \in \overline{\mathcal{G}}_y\,,\\
	E_{\bar{\beta},b|y} &\coloneqq\delta_{b,\beta_y} \sum_c p(\bar{\beta}|c) \, E'_{c,\star|y} \quad \forall b \in \mathcal{G}_y\,,
\end{align}
where
$p(\bar{\beta}|c) \coloneqq \prod_{y:\,\mathcal{G}_y\neq\emptyset} p(\beta_y|y,c)$ (if $\mathcal{G}_y = \emptyset$ for some $y$, we simply set $\beta_y = \emptyset$). It is straightforward to verify that the operators $E_{\bar{\beta},b|y}$ satisfy the conditions~\eqref{appeq:no-signalingSDP}--\eqref{appeq:partial-JM_SDP} whenever the operators $E'_{c,b|y}$ satisfy the conditions~\eqref{appeq:no-signaling}--\eqref{appeq:partial-JM}. Indeed, condition~\eqref{appeq:no-signalingSDP} is satisfied since
\begin{equation}
	\forall y \in  [n]\,, \quad \sum_{b} E_{\bar{\beta},b|y} =\sum_{b\in \overline{\mathcal{G}}_y} E_{\bar{\beta},b|y} + \sum_{b\in {\mathcal{G}}_y} E_{\bar{\beta},b|y} = \sum_c p(\bar{\beta}|c) \left(\sum_{b \in \overline{\mathcal{G}}_y} E'_{c,b|y} +  E'_{c,\star|y}\right)
	 = \sum_c p(\bar{\beta}|c)\, E'_c
	  \eqqcolon E_{\bar{\beta}}\,,
\end{equation}
where we defined $E_{\bar{\beta}} \coloneqq \sum_c  p(\bar{\beta}|c) \, E'_c$.
Condition~\eqref{appeq:consistencySDP} is satisfied due to 
\begin{align}
 &\forall b \in \overline{\mathcal{G}}_y\,, \quad \sum_{\bar{\beta}} E_{\bar{\beta},b|y} = \sum_{\bar{\beta}} p(\bar{\beta}|c) \sum_c E'_{c,b|y} = \sum_c E'_{c,b|y} = B_{b|y} \,,\\
&\forall b \in {\mathcal{G}}_y\,, \quad \sum_{\bar{\beta}} E_{\bar{\beta},b|y} =  \sum_c E'_{c,\star|y}\, \sum_{\bar{\beta}} \delta_{b,\beta_y} p(\bar{\beta}|c)   =  \sum_c E'_{c,\star|y}\,  p(b|y,c) =  \sum_c E'_{c,b|y} = B_{b|y} \,.
\end{align}
Finally, the operators $E_{\bar{\beta},\star|y}$ are given by
\begin{equation}
	E_{\bar{\beta},\star|y} := E_{\bar{\beta}} - \sum_{b \in \overline{\mathcal{G}}_y} E_{\bar{\beta},b|y} = \sum_c  p(\bar{\beta}|c) \left( E'_c -  \sum_{b \in \overline{\mathcal{G}}_y} E'_{c,b|y}\right) = \sum_c  p(\bar{\beta}|c) \, E'_{c,\star|y}\,,
\end{equation}
which leads to condition~\eqref{appeq:partial-JM_SDP}, i.e.,
\begin{equation}
	\forall b \in {\mathcal{G}}_y\,, \quad	E_{\bar{\beta},b|y} = \delta_{b,\beta_y} \sum_c  p(\bar{\beta}|c)\, E'_{c,\star|y} = \delta_{b,\beta_y} E_{\bar{\beta},\star|y}\,.
\end{equation}
This completes the proof.
\end{proof}

\section{Proof of Observation~\ref{obs:3}} \label{app:proof:obs3}
Here we reproduce the proof of Observation~\ref{obs:3} in Ref.~\cite{Masini2024} for completeness.
\begin{obs*}[Masini \emph{et al.}~\cite{Masini2024}]
	If the measurements $\{B_y\}_{y\in  [n]}$ are  $\mathcal{G}$-JM, where $\mathcal{G}_y = [k_y]$ for all $y \neq y'$ and $\mathcal{G}_{y'} = \emptyset$ for a single input $y' \in  [n]$, then the measurements $\{B_y\}_{y \in  [n]}$ are fully jointly measurable.
\end{obs*}
\begin{proof}
Since $\{B_y\}_{y\in [n]}$ are $\mathcal{G}$-JM, the outcomes $\bar{\beta}$ of the quantum instrument $\{\ins_{\bar{\beta}}\}_{\bar{\beta}}$ fully determine the outcomes of all the inputs $y \neq y'$. For the outcomes of the remaining input $y'$, Bob's device performs a measurement with POVM elements $M_{b|y',\bar{\beta}}$ on the post-measurement state. However, this measurement can be performed immediately after the action of the quantum instrument and the resulting classical outcome $b$ can be sent to Bob's device. This procedure does not require any quantum state to be sent to Bob's device. Furthermore, the concatenation of the quantum instrument  ${\{\ins_{\bar{\beta}}\}}_{\bar{\beta}}$ and the subsequent measurement represents a quantum instrument with classical outcomes $(\bar{\beta},b)$, which completely determines the outcomes of all of Bob's measurements. Hence, the measurements $\{B_y\}_{y \in [n]}$ are fully jointly measurable.
\end{proof}

\section{Proof of the bounds~\eqref{eq:bound_qubit_case_d_n2} and~\eqref{eq:bound_qubit_case_d_n}}\label{app:qubit_case_d}

Here we complete the optimizations leading to the bounds
\eqref{eq:bound_qubit_case_d_n2} and
\eqref{eq:bound_qubit_case_d_n}. We start from the general sufficient
condition derived in Sec.~\ref{sec:qubit_case_d},
\begin{equation}\label{appeq:general_finite_n_case_d_bound_app}
	\eta
	\le
	\max_{\|\vec m\|=1,\,0\le\nu<1}
	\min\left\{
		F_{\nu,t_1},
		\min_{y\neq1}G_{\nu,t_y}
	\right\},
\end{equation}
where $F_{\nu,t}=(1-\nu^2)/[2(1-\nu t)]$, $t_y=|\vec m\cdot\vec r_y|$, and
\begin{equation}\label{appeq:G_case_d_app}
	G_{\nu,t}
	=
	\begin{cases}
	1-\nu\sqrt{1-t^2},
	&
	\text{if } \nu\left(t+\sqrt{1-t^2}\right)\le1,
	\\[1ex]
	\dfrac{1-\nu^2}{2(1-\nu t)},
	&
	\text{if } \nu\left(t+\sqrt{1-t^2}\right)\ge1.
	\end{cases}
\end{equation}

Consider first the case of two measurements. We take
$\vec r_1=\hat z$ and
$\vec r_2=\cos(\theta)\,\hat z+\sin(\theta)\,\hat x$, with
$\theta\in[0,\pi/2]$, and choose
$\vec m=\cos(x)\,\hat z+\sin(x)\,\hat x$, with $0\le x\le\theta$.
Then
\begin{equation}
	t_1=\cos(x),\qquad
	t_2=\cos(\theta-x),\qquad
	\sqrt{1-t_2^2}=\sin(\theta-x).
\end{equation}
Using the first branch of $G_{\nu,t_2}$, which is admissible whenever
\begin{equation}\label{appeq:stationary_branch_validity_n2_case_d}
	\nu\left(\cos(\theta-x)+\sin(\theta-x)\right)\le1,
\end{equation}
the two constraints $\eta\le F_{\nu,t_1}$ and
$\eta\le G_{\nu,t_2}$ reduce to
\begin{equation}\label{appeq:optimization_case_d_n2_reduced}
	\eta \le \frac{1-\nu^2}{2(1-\nu\cos(x))}\,,\qquad\eta\le 1-\nu\sin(\theta-x).
\end{equation}
We choose $x=x_*$ such that $\tan(\theta-x_*)=\tan(\theta)/2$, or equivalently,
\begin{equation}\label{appeq:optimal_x_case_d_n2}
	\cos(x_*)=\frac{1+\cos^2(\theta)}{\sqrt{1+3\cos^2(\theta)}}\,,\qquad\sin(x_*)=\frac{\sin(\theta)\cos(\theta)}{\sqrt{1+3\cos^2(\theta)}}\,,
\end{equation}
together with
\begin{equation}\label{appeq:optimal_nu_case_d_n2}
	\nu_*
	=
	\frac{\sqrt{1+3\cos^2(\theta)}}{2+\sin(\theta)}.
\end{equation}
For this choice,
\begin{equation}
	\nu_*\sin(\theta-x_*)=\frac{\sin(\theta)}{2+\sin(\theta)},
	\qquad
	1-\nu_*\sin(\theta-x_*)=\frac{2}{2+\sin(\theta)},
\end{equation}
and a direct substitution also gives
\begin{equation}
	\frac{1-\nu_*^2}{2(1-\nu_*\cos(x_*))}
	=
	\frac{2}{2+\sin(\theta)}.
\end{equation}
Thus, both constraints in Eq.~\eqref{appeq:optimization_case_d_n2_reduced}
are saturated, and we obtain
\begin{equation}
	\eta
	\le
	\frac{2}{2+\sin(\theta)}\,,
\end{equation}
which is Eq.~\eqref{eq:bound_qubit_case_d_n2}.

It remains to check that the first branch of $G_{\nu,t_2}$ is indeed
admissible. Using the above values,
\begin{equation}
	\nu_*
	\left(
		\cos(\theta-x_*)+\sin(\theta-x_*)
	\right)
	=
	\frac{2\cos(\theta)+\sin(\theta)}{2+\sin(\theta)}
	\le1
\end{equation}
for all $\theta\in[0,\pi/2]$. Hence
Eq.~\eqref{appeq:stationary_branch_validity_n2_case_d} holds. The
corresponding optimal value of the postprocessing parameter is
\begin{equation}
	\gamma_2^* = \frac{\nu_*\cos(\theta-x_*)}{1-\nu_*\sin(\theta-x_*)}=\cos(\theta).
\end{equation}

For completeness, we also show that the above choices of $x$ and $\nu$
are optimal, i.e., that for all $0\le\nu\le1$ and $0\le x\le\theta$,
\begin{equation}\label{appeq:n2_case_d_reduced_upper_bound}
	\min\left\{
		\frac{1-\nu^2}{2(1-\nu\cos(x))},
		1-\nu\sin(\theta-x)
	\right\}
	\le
	\frac{2}{2+\sin(\theta)}.
\end{equation}
Let $s=\sin(\theta)$ and $\tau=2/(2+s)$. If $1-\nu\sin(\theta-x)\le\tau$, then the claim is immediate. Otherwise, $\nu\sin(\theta-x)<s/(2+s)$. It is then enough to prove
\begin{equation}
	\frac{1-\nu^2}{2(1-\nu\cos(x))}\le\tau,
\end{equation}
which is equivalent to
\begin{equation}\label{appeq:n2_case_d_quadratic_ineq}
	(2+s)\nu^2-4\nu\cos(x)+2-s\ge0.
\end{equation}
For fixed $x$, the left-hand side is a convex quadratic in $\nu$. Under
the constraint $\nu\sin(\theta-x)\le s/(2+s)$, its minimum is
nonnegative. Indeed, if the unconstrained minimizer
$\nu_0=2\cos(x)/(2+s)$ satisfies the constraint $\nu_0 \sin(\theta-x)\le s/(2+s)$, then
$2\cos(x)\sin(\theta-x)\le s$. This implies $x\ge\theta/2$, and hence
\begin{align}
	(2+s)\nu_0^2-4\nu_0\cos(x)+2-s
	&=
	2-s-\frac{4\cos^2x}{2+s}
	\nonumber\\
	&\ge
	2-s-\frac{4\cos^2(\theta/2)}{2+s}
	=
	\frac{(1-\cos(\theta))^2}{2+s}
	\ge0.
\end{align}
If the unconstrained minimizer does not satisfy the constraint, the
minimum under the constraint is attained at
$\nu=s/[(2+s)\sin(\theta-x)]$. Writing $\beta=\theta-x$, the left-hand
side of Eq.~\eqref{appeq:n2_case_d_quadratic_ineq} becomes
\begin{equation}
	\frac{[2\cos(\theta)\tan(\beta)-\sin(\theta)]^2}
	{[2+\sin(\theta)]\tan^2(\beta)}
	\ge0.
\end{equation}
This proves Eq.~\eqref{appeq:n2_case_d_reduced_upper_bound}.

We now prove Eq.~\eqref{eq:bound_qubit_case_d_n} for an arbitrary
number of measurements $n$. Recall that
\begin{equation}
	\theta
	=
	\max_{y=2,\ldots,n}
	\arccos|\vec r_1\cdot\vec r_y|.
\end{equation}
Choose $\vec m=\vec r_1$. Then $t_1=1$, while for every $y\neq1$,
\begin{equation}
	t_y=|\vec r_1\cdot\vec r_y|\ge\cos(\theta),
	\qquad
	\sqrt{1-t_y^2}\le\sin(\theta).
\end{equation}
The first constraint becomes
\begin{equation}
	\eta\le F_{\nu,1}=\frac{1+\nu}{2}.
\end{equation}
Using the first branch of $G_{\nu,t_y}$ for the other measurements, it
is sufficient to impose the worst-case constraint
$\eta\le1-\nu\sin(\theta)$. Thus
\begin{equation}\label{appeq:cap_case_d_optimization}
	\eta
	\le
	\max_{0\le\nu<1}
	\min\left\{
		\frac{1+\nu}{2},
		1-\nu\sin(\theta)
	\right\}.
\end{equation}
For any fixed $\theta \in [0,\pi/2]$, the first term is affine and increasing in $\nu$, while the second term is affine and monotonically decreasing in $\nu$. Therefore, the optimum is obtained by setting them equal:
\begin{equation}
	\frac{1+\nu_*}{2}=1-\nu_*\sin(\theta)\,.
\end{equation}
This gives
\begin{equation}
	\nu_*=\frac{1}{1+2\sin(\theta)} \quad \text{and} \quad
	\eta
	\le
	\frac{1+\sin(\theta)}{1+2\sin(\theta)},
\end{equation}
which proves Eq.~\eqref{eq:bound_qubit_case_d_n}.

It remains to verify that the first branch of $G_{\nu,t_y}$ is
admissible for all $y\neq1$ at $\nu=\nu_*$. This first branch
corresponds to the critical point
\begin{equation}
	\gamma_y^*
	=
	\frac{\nu t_y}{1-\nu\sqrt{1-t_y^2}}
\end{equation}
obtained by setting to zero the derivative of the denominator of the right-hand side of
Eq.~\eqref{eq:constr2}. This critical point is admissible only if
$0\le\gamma_y^*\le1$. Since $t_y\le1$ and
$\sqrt{1-t_y^2}\le\sin(\theta)$, we have
\begin{equation}
	\gamma_y^*
	\le
	\frac{\nu_*}{1-\nu_*\sin(\theta)}
	=
	\frac{1}{1+\sin(\theta)}
	\le1.
\end{equation}
Thus, the critical point is admissible, completing the proof.

\section{Effect of postselection on the security of QKD protocols} \label{app:postselection}

Below, we present a simple example that illustrates how postselection affects the information available to Eve in QKD protocols.

Consider a scenario where Alice and Bob share the maximally entangled state
$|\Phi\rangle = \frac{1}{\sqrt{d}} \sum_{i=1}^d |i\rangle_A |i\rangle_B,$ and perform measurements in the computational basis to generate a raw key. We assume Bob's detector is affected by loss as in Section~\ref{sec:applications}. Consider the attack described in case \emph{(c)} or \emph{(d)} of Section~\ref{sec:generic_strategies}, which gives Eve full information about Bob's conclusive outcomes with probability $\eta$, for a suitable $0<\eta<1$. However, with probability $1-\eta$, Bob obtains a no-click outcome $\varnothing$; in these instances, Eve's measurement yields a random outcome in $\{1, \dots, d\}$. 
Accordingly, Alice, Bob, and Eve share random variables $A \in \{1, \dots, d\}$, $B \in \{1, \dots, d, \varnothing\}$, and $E \in \{1, \dots, d\}$ with joint probability distribution 
\begin{equation}\label{eq:example_postselection}
    P(A=a, B=b, E=e) = 
    \begin{cases} 
        \frac{\eta}{d} & \text{if } a=b=e \text{ and } b \neq \varnothing, \\
        \frac{1-\eta}{d^2} & \text{if } b = \varnothing.
    \end{cases}
\end{equation}
In other words, whenever Bob obtains a conclusive outcome, Alice and Eve are perfectly correlated with Bob, while whenever Bob obtains a no-click outcome, Alice and Eve are uncorrelated with Bob and with each other. 

Suppose Alice and Bob repeat this process for $n$ trials, resulting in the raw strings $\mathsf{A} = [{A}_1, \dots, {A}_n]$, $\mathsf{B} = [{B}_1, \dots, {B}_n]$, and $\mathsf{E} = [{E}_1, \dots, {E}_n]$. Let $\mathsf{C} = [{C}_1, \dots, {C}_n]$ be the string of postselection flags, where ${C}_i = 1$ if ${B}_i \neq \varnothing$ and ${C}_i = 0$ if ${B}_i = \varnothing$. 

To illustrate the flaw in Ref.~\cite{Branciard2012}, let us compute $H(\mathsf{A}|\mathsf{E})$ and $H(\mathsf{A}|\mathsf{E},\mathsf{C})$ i.e., the conditional Shannon entropy of $\mathsf{A}$ given Eve's side information $\mathsf{E}$ and $\mathsf{E}'=(\mathsf{E},\mathsf{C})$. For a single round, the joint distribution of $A$ and ${E}$ is
\begin{equation}
	P(A=a,E=e) = 
	\begin{cases} 
		\frac{\eta}{d} + \frac{1-\eta}{d^2} & \text{if } a=e, \\
		\frac{1-\eta}{d^2} & \text{if } a \neq e.
	\end{cases}
\end{equation}
It then follows that 
\begin{equation}
H(\mathsf A|\mathsf E)
=n\left[-\left(\eta+\frac{1-\eta}{d}\right)
\log_2\left(\eta+\frac{1-\eta}{d}\right)
-(d-1)\frac{1-\eta}{d}
\log_2\left(\frac{1-\eta}{d}\right)\right].
\end{equation}
Let us now compute $H(\mathsf{A}|\mathsf{E}') = H(\mathsf{A}|\mathsf{E},\mathsf{C})$. If $C_i=1$, then $A_i = B_i =E_i$, so $H(A_i|E_i,C_i=1)=0$. If $C_i=0$, then $A_i$ is uniformly distributed in $\{1, \dots, d\}$ and independent of $E_i$, so $H(A_i|E_i,C_i=0)=\log_2(d)$. Since $P(C_i=1) = \eta$ and $P(C_i=0) = 1-\eta$, we have
\begin{equation}
H(\mathsf{A}|\mathsf{E}') = n\left[\eta\cdot 0 + (1-\eta)\log_2(d)\right] = n(1-\eta)\log_2(d).
\end{equation}
It can be verified that $H(\mathsf{A}|\mathsf{E}')<H(\mathsf{A}|\mathsf{E})$ for $0<\eta<1$.
Intuitively, before Bob's announcement, Eve's variable $\mathsf{E}$ conflates two regimes: click rounds, where she has full knowledge of $\mathsf{A}$, and no-click rounds, where she knows nothing about $\mathsf{A}$. The announcement allows Eve to separate these two regimes, strictly increasing her knowledge about $\mathsf{A}$. 
In the context of the security proof of Ref.~\cite{Branciard2012}, this implies that it is not justified to assume that Eve's information is upper bounded by $H(\mathsf{A}|\mathsf{E})$ when Bob postselects on the conclusive events, as Eve's information can be strictly larger than this quantity.

When the attack described above is applied to a QKD protocol and Alice, Bob, and Eve share the random variables $\mathsf{A}$, $\mathsf{B}$, and $\mathsf{E}$ with the distribution in Eq.~\eqref{eq:example_postselection}, Alice and Bob can clearly not distill a secret key if they discard the no-click events, as Bob's string $\mathsf{B}$ is perfectly correlated with Eve's string $\mathsf{E}$.
However, if Alice and Bob do not discard the no-click events, they can distill a secret key, even though (i) Eve has full information about Bob's conclusive outcomes and (ii) Bob's no-click events are uncorrelated with the values held by Alice and Eve. Indeed, by Csisz\'ar-K\"orner's formula~\cite{Csiszar1978}, a positive key rate can be established with direct reconciliation (Alice sends information to Bob and Bob tries to correct his data) if $I(\mathsf{A}:\mathsf{B})-I(\mathsf{A}:\mathsf{E}) > 0$ and with reverse reconciliation (Bob sends information to Alice and Alice tries to correct her data) if $I(\mathsf{B}:\mathsf{A})-I(\mathsf{B}:\mathsf{E}) > 0$, where $I(X:Y)$ denotes the mutual information of $X$ and $Y$. A direct computation using the probability distribution in Eq.~\eqref{eq:example_postselection} gives
\begin{equation}
	I(\mathsf{B}:\mathsf{A})-I(\mathsf{B}:\mathsf{E})=0
\end{equation}
implying that Bob cannot initiate the reconciliation process to distill a key, as Eve's knowledge of Bob's outcomes is identical to Alice's.
However, a direct computation also gives
\begin{equation}
	I(\mathsf{A}:\mathsf{B})-I(\mathsf{A}:\mathsf{E}) = \eta \log_2(d)-\frac{1}{d}\left[\left(1+\eta(d-1)\right)\log_2\left(1+\eta(d-1)\right)+(d-1)(1-\eta)\log_2(1-\eta)\right]\,
\end{equation}
which is positive for all $0<\eta<1$. Hence, a positive key rate can be established with direct reconciliation, even though Eve has full information about Bob's conclusive outcomes and the no-click events are uncorrelated with Alice and Eve.

To illustrate this, consider the following example strings (assuming $d=2$ for simplicity) that satisfy the conditions of the model:
\begin{equation}
\begin{aligned} 
    \mathsf{A} &= 001011010100010\dots \\ 
    \mathsf{B} &= 001\varnothing11010\varnothing00\varnothing10\dots \\ 
    \mathsf{E} &= 001111010000110\dots 
\end{aligned}
\end{equation}
Via public communication, Bob asks Alice to compute the modulo-2 sum (parity) of her first five bits. Alice responds that the sum is $0$. Since Bob knows his own first five bits are $(0, 0, 1, \varnothing, 1)$, the parity check reveals that Alice's $4^{\text{th}}$ bit must be $0$. Bob updates his string accordingly:
\begin{equation}
    \mathsf{B}' = 001011010\varnothing00\varnothing10\dots
\end{equation}
This same public communication informs Eve that at least one of her first five bits differs from Alice's, as the sum of Eve's first five bits is $1$. However, because Eve does not know the exact position of the erroneous bit, she cannot correct her string with certainty.
By repeating this procedure, Alice and Bob can arrive at perfectly correlated raw strings, while Eve remains imperfectly correlated, allowing them to distill a secret key.

\end{document}